%% file: ms.tex
\documentclass[11pt]{article}
\pdfoutput=1

\usepackage[english]{babel}
\usepackage[utf8x]{inputenc}
\usepackage[T1]{fontenc}
\usepackage[letterpaper,top=1in,bottom=1in,left=1in,right=1in,marginparwidth=1.75cm]{geometry}
\usepackage{dsfont}
\usepackage{amsmath, amssymb, amsthm, verbatim,enumerate,bbm}
\usepackage{indentfirst, bbm}
\usepackage{appendix}
\usepackage{enumerate}
\usepackage{tikz}
\usepackage{verbatim}
\usepackage[colorinlistoftodos]{todonotes}
\usepackage[colorlinks=true, allcolors=blue]{hyperref}
\usepackage[nameinlink,capitalise,noabbrev]{cleveref}
\usepackage{subfig}
\usepackage{algorithm,algpseudocode}
\usepackage{enumitem}
\crefname{appsec}{Appendix}{Appendices}
\creflabelformat{enumi}{#2textup{#1}#3}

\usepackage{thm-restate}
\usepackage{thmtools}
\usepackage{wrapfig}

\allowdisplaybreaks
\makeatletter
\newtheorem{theorem}{Theorem}[subsection]
\renewcommand{\thetheorem}{\ifnum\value{subsection}=0 \thesection\else\thesubsection\fi.\arabic{theorem}}

\newtheorem{lemma}[theorem]{Lemma}

\newtheorem{claim}[theorem]{Claim}

\newtheorem{corollary}[theorem]{Corollary}

\newtheorem*{question*}{Question}

\theoremstyle{definition}
\newtheorem{definition}[theorem]{Definition}

\newtheorem{remark}[theorem]{Remark}

\newtheorem{example}{Example}
\theoremstyle{plain}

\definecolor{jgreen}{RGB}{0, 153, 51}
\makeatother

{\makeatletter
	\gdef\xxxmark{%
		\expandafter\ifx\csname @mpargs\endcsname\relax 
		\expandafter\ifx\csname @captype\endcsname\relax 
		\marginpar{xxx}
		\else
		xxx 
		\fi
		\else
		xxx 
		\fi}
	\gdef\xxx{\@ifnextchar[\xxx@lab\xxx@nolab}
	\long\gdef\xxx@lab[#1]#2{{\bf [\xxxmark #2 ---{\sc #1}]}}
	\long\gdef\xxx@nolab#1{{\bf [\xxxmark #1]}}
}

\makeatletter
\let\orgdescriptionlabel\descriptionlabel
\renewcommand*{\descriptionlabel}[1]{%
  \let\orglabel\label
  \let\label\@gobble
  \phantomsection
 \edef\@currentlabel{#1\unskip}%
  \let\label\orglabel
  \orgdescriptionlabel{#1}%
}
\makeatother

\title{Motif Cut Sparsifiers}
\author{Michael Kapralov\\EPFL \and Mikhail Makarov\\EPFL \and Sandeep Silwal\\MIT \and Christian Sohler\\University of Cologne \and Jakab Tardos\\EPFL}

\newcommand{\tO}{\widetilde{O}}

\begin{document}
\pagenumbering{gobble}
\maketitle

\begin{abstract}

A \emph{motif} is a frequently occurring subgraph of a given directed or undirected graph $G$~\cite{milo2002}. Motifs capture higher order organizational structure of $G$ beyond edge relationships, and, therefore, have found wide applications such as in graph clustering, community detection, and analysis of biological and physical networks to name a few~\cite{benson2016,tsourakakis2017}. In these applications, the cut structure of motifs plays a crucial role as vertices are partitioned into clusters by cuts whose conductance is based on the number of instances of a particular motif, as opposed to just the number of edges, crossing the cuts. 

In this paper, we introduce the concept of a motif cut sparsifier. We show that one can compute in polynomial time a sparse weighted subgraph $G'$ with only $\tO(n/\epsilon^2)$ edges such that for every cut, the weighted number of copies of $M$ crossing the cut in $G'$ is within a $1+\epsilon$ factor of the number of copies of $M$ crossing the cut in $G$, {\em for every constant size motif $M$}.

Our work carefully combines the viewpoints of both graph sparsification and hypergraph sparsification. We sample edges which requires us to extend and strengthen the concept of cut sparsifiers introduced in the seminal work of~\cite{karger1999random} and~\cite{benczur_randomized_2015} to the motif setting. The task of adapting the importance sampling framework common to efficient graph sparsification algorithms to the motif setting turns out to be nontrivial due to the fact that cut sizes in a random subgraph of $G$ depend non-linearly on the sampled edges. To overcome this, we adopt the viewpoint of hypergraph sparsification to define edge sampling probabilities which are derived from the strong connectivity values of a hypergraph whose hyperedges represent motif instances. Finally, an iterative sparsification primitive inspired by both viewpoints is used to reduce the number of edges in $G$ to nearly linear. 

In addition, we present a strong lower bound ruling out a similar result for sparsification with respect to {\em induced} occurrences of motifs.

\end{abstract}

\newpage

\global\long\def\R{\mathbb{R}}

\global\long\def\S{\mathbb{S}}

\global\long\def\Z{\mathbb{Z}}

\global\long\def\C{\mathbb{C}}

\global\long\def\Q{\mathbb{Q}}


\global\long\def\P{\mathbb{P}}

\global\long\def\F{\mathbb{F}}

\global\long\def\U{\mathcal{U}}

\global\long\def\V{\mathcal{V}}

\global\long\def\E{\mathbb{E}}

\global\long\def\Ev{\mathscr{Rk}}

\global\long\def\Dg{\mathscr{D}}

\global\long\def\Ndg{\mathscr{ND}}

\global\long\def\Rv{\mathcal{R}}

\global\long\def\Gv{\mathscr{Null}}

\global\long\def\Hv{\mathscr{Orth}}

\global\long\def\Supp{{\bf Supp}}

\global\long\def\Sv{\mathscr{Spt}}

\global\long\def\ring{\mathfrak{R}}

\global\long\def\Bad{{\boldsymbol{B}}}

\global\long\def\supp{{\bf supp}}

\global\long\def\A{\mathcal{A}}

\global\long\def\L{\mathcal{L}}

\newcommand{\event}{\mathcal E}

\newcommand{\p}{\mathbb P}

\newcommand{\eps}{\varepsilon}

\newcommand{\1}{\boldsymbol 1}
\newcommand{\N}{\mathcal{N}}

\newcommand{\e}{\epsilon}
\newcommand{\M}{M}
\newcommand{\FG}{F}

\newcommand{\ws}{\eta}
\newcommand{\wse}{\widehat{\eta}}
\newcommand{\wws}{\mu}
\newcommand{\wwse}{\nu}
\newcommand{\wwsee}{\widehat{\nu}}

\newcommand{\HI}{\mathbf{H}}

\newcommand{\Val}{\textup{Val}}
\newcommand{\IVal}{\widetilde{\textup{Val}}}

\newcommand{\parspar}[0]{\textsc{PartialSparsification}}
\newcommand{\gparspar}[0]{\textsc{GeneralPartialSparsification}}
\newcommand{\fparspar}[0]{\textsc{FastPartialSparsification}}
\newcommand{\motspar}{\textsc{MotifSparsification}}
\newcommand{\MG}[2]{\mathcal{M}(#1, #2)}

\newcommand{\OVal}{\overline{\textup{Val}}}
\newcommand{\OM}[2]{\overline{\mathcal M}(#1, #2)}
\newcommand{\wh}{\widehat}

\renewcommand{\todo}[1]{\xxx{#1}}

\tableofcontents

\newpage

\pagenumbering{arabic}

\input{intro}

\input{related_work}

\input{tech_overview}
\input{prelims}

\input{notations}
\input{algorithm}
\input{analysis_partial}
\input{analysis_motif}
\input{sublinear_intro}
\input{motif_weighted_graph}
\input{fast_partial_sparsification}
\input{lower_bound}

\section*{Acknowledgments}
Mikhail Makarov and Jakab Tardos are supported  by ERC Starting Grant 759471. Michael Kapralov is supported in part by ERC Starting Grant 759471. Sandeep Silwal is supported by an NSF Graduate Research Fellowship under Grant No. 1745302, NSF TRIPODS program (award DMS-2022448), and Simons Investigator Award.

\bibliographystyle{alpha}
\bibliography{main}

\input{appendix}

\end{document}

%% file: intro.tex
\section{Introduction}

A motif is a (connected) subgraph of a given directed or undirected graph $G=(V,E)$ that occurs more frequently than one would typically assume in a random graph; it has been observed empirically that motifs exist in many networks~\cite{milo2002, yaverolu2014, benson2016, tsourakakis2017}. These higher order graph structures are crucial to the organization of complex networks as they capture richer structural information about the graph data and therefore
carry important information that can be exploited in network data analysis. Indeed, in many application domains, such as in clustering and social network analysis~\cite{satuluri2011, benson2016, lipan2017,tsourakakis2017,hao2017, li2017, lipan2019}, community detection~\cite{satuluri2011, benson2016, hao2017,tsourakakis2017,paranjape2017, seshadhri2020,sotiropoulos2021}, and analysis of biological or physical networks~\cite{mangan2003, wasserman2007, wong2011, benson2016}, understanding higher order graph structures has become increasingly important. See \cref{sec:related_works} for further details on motif-based applications. 

Graph clustering in particular is a prominent example where clustering algorithms have been developed to exploit the motifs structure of graphs~\cite{benson2016, tsourakakis2017}. These algorithms first compute a motif weighted graph where every edge is weighted by the number of copies of a given motif it
is contained in, and then apply spectral clustering on this motif weighted graph (see \cref{sec:related_works} for more details). Such an approach may be viewed as partitioning the vertex set of a graph into subsets (called clusters) with high internal motif connectivity and low motif connectivity 
between the clusters. 

Graph sparsification is an algorithmic technique for speeding up cut based graph algorithms that was introduced in the seminal work of~\cite{karger1999random} and~\cite{benczur_randomized_2015}, with powerful generalization to spectral sparsifiers obtained in~\cite{DBLP:journals/siamcomp/SpielmanT11}. The main idea behind graph cut sparsification is to design a sparse weighted graph that approximates the cuts in the original graph to within a $1\pm \e$ factor for small $\e\in (0, 1)$. Cut sparsifiers  with $\widetilde{O}(n/\e^2)$ edges that approximate all cuts in $G$ have been constructed, with some constructions achieving an $O(n/\e^2)$ upper bound on the number of edges in nearly linear time~\cite{benczr1996}. The related concept of hypergraph sparsification has received a lot of attention in the literature recently, with nearly optimal size sparsifiers obtained in~\cite{DBLP:journals/corr/abs-2009-04992}. In this paper we ask whether it is possible to sparsify a graph while preserving the motif cut structure:

\begin{center}
 Given an arbitrary input graph $G$, is it possible to compute a sparse weighted graph $G'$ (a motif cut sparsifier) that approximates the motif cut structure of $G$?
\end{center}


Before we discuss  how motif sparsification compares to graph and hypergraph sparsification, we first informally state  our definition of a motif cut sparsifier. The main idea is very intuitive: a motif cut sparsifier approximates the number of motifs that cross a cut for every cut in the graph. In order to utilize sparse graphs, edges need to be weighted and we must define the weighted number of motifs crossing a cut. Here we follow the standard interpretation of integer edge weights as edge multiplicities, and therefore, define the motif weight as the product of its edge weights (which under the previous interpretation is simply the number of distinct unweighted motifs crossing the cut). The definition generalizes to non-integral edge weight in a straightforward manner.

\begin{definition}[Motif cut sparsifier; informal]
For a connected motif $M$ and $\e\in (0, 1)$ we say that  a (possibly directed) weighted graph $G'=(V, E)$ is an $\e$-motif-sparsifier of $G$ with respect to $M$ if for every $\emptyset \neq S \subset V$ the weighted number of copies of $M$ in $G$ crossing the cut $(S, V\setminus S)$ is $(1\pm \e)$-close to the number of copies of $M$ crossing the same cut in $G'$. 
\end{definition}

There is no consensus in the literature on whether these "copies" should be \emph{induced} subgraphs of $G$ or \emph{arbitrary} subgraphs -- both seem to be useful concepts in applications. We consider both cases, and it turns out there is a fundamental difference between them: In the case of non-induced motifs powerful and small motif-cut sparsifiers can be constructed for any graph $G$ (as we'll see below) while in the case of induced motifs this is not possible. Hence, below we focus on the non-induced case, and we state our result for the induced case at the end of the section.

\paragraph{Motif sparsifiers vs hypergraph sparsifiers.} It may seem at first sight that one can easily compute a motif cut sparsifier by first computing a motif hypergraph that contains an edge for every motif, and then by sparsifying this hypergraph. The issue with this approach is that although there exists a corresponding motif hypergraph for every graph and every motif (at least when we allow parallel hyperedges), the converse is not true. Thus, while we can compute a motif hypergraph sparsifier, we do not know how to transform it back into a graph while maintaining the fact that the number of motifs crossing every cut is preserved. Similar issues arise if we first sparsify a motif weighted graph. This is illustrated in \cref{fig:naive_graph} and detailed in \cref{sec:technical_overview} . 

Indeed, motif sparsifiers are quite different from graph and hypergraph cut sparsifiers. For example, graph and hypergraph cut sparsifiers have the property that when $G'=(V,E')$ is a sparsifier of $G(V,E)$ and $H'=(V,F')$ is a sparsifier for $H(V,F)$ then $(V,E' \cup F')$ is a sparsifier for $(V,E\cup F)$. This property can, for example, be used to obtain a semi-streaming algorithm for many cut problems using $O(n \cdot \text{poly}(\log n))$ space~\cite{ahn2009, kapralov2014, rubinstein2018, assadi2019, mukhopadhyay2020, assadi2021}.

Unfortunately, motif sparsifiers in general do not have this property. Furthermore, even for a small motif like a triangle, it is not possible to compute a motif sparsifier in the semi-streaming model. This is because even counting the number of triangles in a stream can require $\Omega(|E|)$ space for $|E| = \Omega(n^2)$~\cite{braverman2013} and computing a motif sparsifier, in particular when the motif is a triangle, easily allows us to recover the global triangle count by querying the sparsifer on the $n$ singleton cuts.

\paragraph{Importance sampling.} A common approach to different graph and hypergraph sparsification algorithms (see ~\cite{benczr1996,benczur_randomized_2015,newman2013,kogan2015sketching,soma2019,kapralov2021, fung2019general} and references within) is to define a sampling probability $p(e)$ and a weight $w(e)$ for each edge $e$ and then sample each edge independently with probability $p(e)$. If $e$ is sampled, it is also assigned weight $w(e)$; for appropriately defined probabilities and weights, the resulting graph is a sparsifier with a near linear number of edges. 

For motif sparsifiers, such an approach \emph{cannot} yield a cut sparsifier of near linear size, as the example of a clique on $n$ vertices with the motif being a triangle shows. Indeed, if we sample every edge with probability $o(1/n^{2/3})$, then the expected number of triangles incident to a given vertex is $o(1)$. Then, it is straightforward to show that the resulting graph is typically not a triangle sparsifier. However, for a sampling probability of 
$\Omega(1/n^{2/3})$, the expected number of sampled edges is 
$\Omega(n^{4/3})$, i.e. the resulting graph does not have near linear size. Since a clique is also completely symmetric, it is unclear how one could assign different probabilities to each edge. However, there is still a simple argument that a sampling probability of roughly $p=\log n / n^{2/3}$ results in a sparsifier such that w.h.p.\,no vertex is incident to more than $\log^{O(1)} n$ distinct triangles. Since every triangle has three edges, this implies that there are only $n \log^{O(1)} n$ edges that are involved in a triangle. Thus, removing the remaining edges yields a triangle sparsifier of near linear size.

While our construction still samples every edge with the same probability, in the special case of a clique, we can only obtain a sparsifier if we remove most of the unused edges in a cleaning step. It is unclear whether such an approach generalizes to other less structured graphs and motifs. Nevertheless, the main result of this paper is that there does exist an algorithm producing a motif sparsifier of nearly linear size from an arbitrary input graph:

\begin{theorem}[follows from \cref{cor-main} and \cref{thm:intro-main-general-2} in  \cref{sec:main-results}]\label{thm:intro-main}
For every graph $G=(V, E)$, $|V|=n$, every constant integer $r\geq 2$, and $\e \in (0, 1)$, there exists an $\e$-motif sparsifier $G'$ of $G$ with respect to {\bf all connected motifs $M$ of with at most $r$ vertices simultaneously}  that contains  $\tO(n/\e^2)$ edges. 

Furthermore, there is an algorithm which outputs a $G'$ which is an $\e$-motif sparsifier with high probability. Its running time is $\tO(\min(T(r), n^{\omega\lceil r / 3 \rceil}))$, where $T(r)$ is the time need to enumerate all of the motif instances and $n^\omega$ is the matrix multiplication time.
\end{theorem}

Note that the resulting graph $G'$ is automatically a cut sparsifier of $G$, as an $M$-sparsifier is exactly a cut sparsifier when $M$ is a single edge. Beyond that, however, $G'$ approximately preserves the sizes of {\em all motif cuts} in $G$ with respect to constant size motifs. Theorem~\ref{thm:intro-main} also applies to directed graphs.

The running time -- $\widetilde{O}(n^{\omega\lceil r/3\rceil})$ in particular -- is sublinear in the number of motif instances in some settings. This shows a clear advantage of motif sparsification over simply sparsifying the motif hypergraph, which would take time at least proportional to the number of hyperedges (ie. motif instances).

\paragraph{Induced Motifs.}
In the final section of the paper we consider the setting where we require motif instances to be \emph{induced} subgraphs of input graph $G$. This is also a natural definition of motifs which likewise has been extensively studied in literature; see~\cite{alon2008,tsourakakis2017,bressan2021, bressan2021_2} and the references within. We show that no analogue of \cref{thm:intro-main} exists in this setting. Even for constant size motifs we can construct an example where any non-trivial sparsification is impossible.

\begin{theorem}[Informal version of Theorem~\ref{thm:graphlet-main}]\label{thm:intro-graphlet-main}
There exists a graph $G=(V, E)$ on $n$ vertices and a motif of constant size such that it is impossible to approximate the induced-motif-cut structure of $G$ to within a multiplicative error of $(1\pm\epsilon)$ for $\epsilon\le1/500$ using a 
(non-negative) weighted graph with $o(n^2)$ edges.
\end{theorem}

%% file: related_work.tex
\subsection{Related Work}\label{sec:related_works}

As stated in the introduction, motifs have been widely adopted for study of higher order networks due to their ubiquitous presence~\cite{milo2002, yaverolu2014, benson2016}. Since the network literature concerning motifs is too vast to properly summarize, we mainly focus on algorithms and applications of motifs and higher order structures. Note that a majority of the papers we reference are application oriented papers; relatively few works offer strong theoretical guarantees.

Applications where motif analysis has become impactful include graph clustering (both local and global clustering)~\cite{satuluri2011, benson2016, lipan2017,tsourakakis2017, lipan2019} and community detection~\cite{satuluri2011, benson2016, hao2017,tsourakakis2017,paranjape2017, seshadhri2020,sotiropoulos2021}. These applications are based on exploiting the motif-cut structure of a given graph. For example in works such as~\cite{benson2016, hao2017, tsourakakis2017}, various alternative notions of conductance are introduced which take into account the influence of motifs. In particular, the definition of conductance is redefined in terms of the number of motifs, for example triangles, crossing the cut. Therefore, one direct application of our results is to provide solid theoretical understanding of motif-based cut structure via graph sparsification.

In graph and network data visualization, it has been empirically observed that motif based embeddings provide more meaningful low-dimensional representations over their counterparts which do not employ motifs, such as spectral embeddings~\cite{zhang2018,nassar2020}. Indeed,~\cite{nassar2020} shows that performing spectral emebeddings on adjacency matrices which are motif based, for example using matrices which are weighted sums of higher powers of the adjacency matrix, leads to better inductive bias as these presentations better capture the rich underlying community or cluster structures; see the visualizations given in~\cite{zhang2018,nassar2020}.

In graph classification, motifs have provided more meaningful characterizations for graphs at both micro (local) and macro (global) scales~\cite{ahmed2016}. Motifs have also become popular in the related area of learning on graphs which has further downstream applications such as recommender systems, fraud detection, and protein identification~\cite{rossi2018, eswaran2020, tudisco2021}. Additional applications of motif-based graph learning include link prediction~\cite{benson2018,arrigo2020, rossi2020} and computing network-based node rankings~\cite{benson2019_hypergraph,arrigo2020}. Indeed in the active area of graph neural networks, motif counts are an extremely popular feature augmentation technique as graph neural networks often struggle to identify motifs and higher order structures~\cite{xu2019, zhao2021, liu2021}.

Lastly, there has also been empirical and theoretical work on efficiently counting motifs and summarizing motif statistics. This literature is also quite vast but an excellent reference is the tutorial~\cite{seshadhri_tutorial} given at the WWW 2019 conference.

Note that which motifs are important for a given complex network strongly depends on the underlying network properties~\cite{milo2002, mangan2003, benson2016}. One of the most fundamental and well studied motifs is the triangle and its directed variants~\cite{tsourakakis2011, satuluri2011, benson2016, tsourakakis2017,  seshadhri2020}. Indeed, some of the work closest to ours concerns triangle motifs. 

Objects close to triangle sparsifiers, which we precisely define and give theoretical guarantees in our work, have also been studied~\cite{tsourakakis2011, sotiropoulos2021}. The main difference is that in~\cite{tsourakakis2011}, their goal is to acquire a sparse subgraph which only preserves the \emph{global} triangle count; in contrast, our task is much more difficult as we wish to preserve the triangle counts (and arbitrary motif counts) \emph{for all cuts} simultaneously. Note that preserving motif cut values automatically implies preservation of the global number of triangles by querying $n$ singleton cuts. Furthermore,~\cite{tsourakakis2011} employ a one-shot uniform sampling of the edges whereas we use careful importance-based sampling based on edge importance over multiple rounds. Similarly in~\cite{sotiropoulos2021}, their goal is to get a sparsifier with respect to edges which has better space bounds for graphs containing many triangles. Our work achieves nearly linear space bounds for preserving motifs cuts for arbitrary motifs.

\paragraph{Clique enumeration results.}
Our first algorithm makes use of a primitive that enumerates all of the instances of a given motif. Unfortunately in general, this can take time exponential in the size of the motif, since even deciding if certain motifes are contained in a graph, such as a clique, is NP-complete \cite{karp1972reducibility}. 

The clique enumeration problem is one of the most studied motif enumeration problems. The most notable results here include \cite{chiba1985arboricity}, giving an algorithm working in time $O(r \alpha(G)^{r - 2} m)$, where $\alpha(G)$ is the arboricity of the graph $G$ for enumerating all cliques of size $r$. By utilizing the bound $\alpha(G) \leq m^{1/2}$ for connected graphs from the same paper, this yields an $O(m^{r/2})$ time algorithm for a general graph. 

There are also works which achieve faster runtimes for graph enumeration for subgraphs with special structures, such as planar graphs or bounded tree-width graphs \cite{alon1995color}, and bounded arboricity graphs \cite{chiba1985arboricity}. Lastly, see \cite{ribeiro2021survey} and references within for a survey on applied algorithms for subgraph enumeration.

%% file: tech_overview.tex
\section{Technical Overview}\label{sec:technical_overview}

We illustrate our main algorithmic ideas by considering a simple example, namely when $G=(V, E)$ is an undirected unweighted graph and the motif $M$ is a triangle $\Delta=(V_\Delta, E_\Delta)$, i.e. a clique on three vertices. Our approach is inspired by the techniques introduced by Karger~\cite{karger1999random} and Benczur and Karger~\cite{benczur_randomized_2015} in the context of  sparsification of undirected graphs. We recall these techniques now, then show why their immediate extension fails, and finally present our algorithm.

We start by recalling Karger's cut sampling bound and its application to graph sparsification. Karger~\cite{karger1999random} shows that in a graph $G$ with min-cut $k$, the number of cuts of size $\alpha k$ for $\alpha\geq 1$ is bounded by ${n \choose 2\alpha}$. The bound is then applied to show that a sample of edges of $G$ which contains every edge independently with probability $p=\min\left\{\frac{C\log n}{\e^2 k}, 1\right\}$ (with weight $1/p$) is an $\e$-cut sparsifier, i.e. preserves all cuts up to multiplicative precision $1\pm \e$, with high probability. The latter claim follows by noting that the probability that a cut of size $\alpha k$ is not appropriately preserved is exponential in $\e^2 p \alpha k=\Omega(C\alpha \log n)$, which suffices for the union bound.  This uniform sampling approach leads to a reduction in the number of edges when the min-cut $k$ in $G$ is large. In the general case~\cite{benczur_randomized_2015} show that sampling edges with probabilities proportional to the inverse of their strong connectivity and reweighting appropriately leads to a cut sparsifier with high probability. Here  the strong connectivity $k_e$ of an edge $e$ is equal to the maximum $k$ such that there exists a vertex induced subgraph $C$ of $G$ containing $e$ such that the size of its minimum cut in $C$ is at least $k$. 

In what follows we discuss two natural approaches to using these techniques to obtain motif sparsifiers, explain some of the issues with them, and then outline our approach. The first approach is based on a hypergraph version of motifs and the second one is based on graphs.
 In the following discussion we assume for simplicity that the input graph $G$ is undirected and unweighted and the motif $M$ is a triangle. 

\paragraph{Motif sparsification based on hypergraphs?}

As already mentioned in the introduction
one can compute a motif hypergraph by creating a hyperedge for every motif. We could then simply use hypergraph cut sparsification algorithms, such as \cite{DBLP:journals/corr/KoganK14} or \cite{DBLP:journals/corr/abs-2009-04992}. Although in general, we cannot transform a sparsified hypergraph back into a graph, we could still try to adapt some
hypergraph sparsification techniques to our problem. For example, we could sample all edges of a motif whenever its corresponding hyperedge gets picked. 
To give a concrete example, in the case of triangle motifs, we may first find all triangles in the input graph, select some of them and then construct a new graph, containing only the selected triangles with some edge re-weightings. This would be a way to simulate some hypergraph sparsification approaches. However,
it is easy to see that some of the discarded triangles might appear again. For example, consider a case of the graph on \cref{fig:triangle_graph}: if you take only triangles $1$, $2$ and $3$ and reconstruct the graph, the final graph will still contain triangle $4$.
Therefore, we cannot hope to directly transform hypergraph sparsification approaches into motif sparsifiers.

\begin{wrapfigure}{r}{0.4 \textwidth}
    \centering
    \includegraphics[scale=0.7]{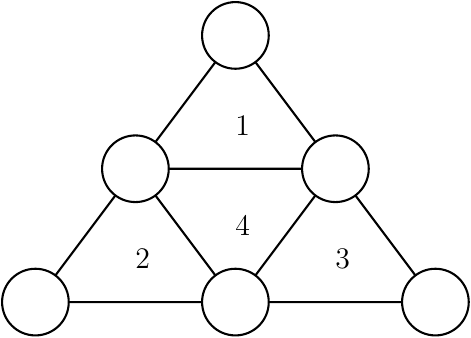}
    \caption{Sampling triangles does not lead to triangle sparsification.}
    \label{fig:triangle_graph}
\end{wrapfigure}

\paragraph{Motif sparsification based on a triangle-weighted graph?}
A natural way to apply Karger's approach to our motif sparsification problem (or triangle sparsification in the following discussion) is to use it on the {\em triangle weighted graph}  $G_\Delta=(V, E, w_\Delta)$, where $w_\Delta(e)$ is the number of triangles containing edge $e$ that has been used in the context of graph clustering~\cite{benson2016,tsourakakis2017}. Indeed, triangle weighted graphs have the useful property that the size of the cut $(S, V \setminus S)$ in $G_\Delta$ is exactly twice the number of triangles that cross this cut in $G$. Therefore, if we were to sparsify $G$ to $G'$ in such a way that $G'_\Delta$ is a cut sparsifier of $G_\Delta$, $G'$ would be a motif cut sparsifier of $G$. It is a seemingly natural approach to try to use triangle weighted graphs to obtain triangle sparsifiers. 
However, we will now show in a series of examples that a number of simple approaches which use the triangle weighted graph fail.

A  naive approach using the triangle weighted graph would be to sparsify the triangle weighted graph, and then construct $G'$ by taking the remaining edges in $G'_\Delta$ with some weights. However, this does not work, as a situation could easily arise where all of the triangles in some cuts are deleted. Consider the example in \cref{fig:naive_graph}.
\begin{figure}[h]
    \centering
    \includegraphics{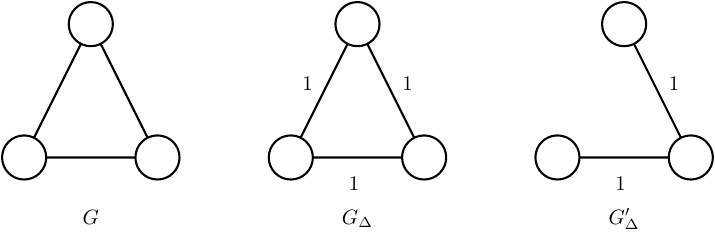}
    \caption{Sparsifying triangle weighted graph does not work for triangle sparsification.}
    \label{fig:naive_graph}
\end{figure}
Here, $G'_\Delta$ is clearly a $1/2$-cut sparsifier of $G_\Delta$, but since it contains no triangles, no motif sparsifier of $G$ can be constructed from it without adding new edges.  

A better approach is  to apply Karger's cut counting bound to the triangle weighted graph $G_\Delta$ and use it to prove that an appropriate random sample of edges of $G$, denoted by $G'$, will satisfy 
\begin{equation}\label{eq:gdgpd}
G_\Delta \approx_{\e} G'_\Delta.
\end{equation}
First, in order to make this approach work, we need to assume that $G_\Delta$ is $k$-connected for some reasonably large $k$. We make this assumption now to illustrate the challenges that arise even in this special case. Following \cite{karger1999random}, we could sample each edge with probability $\approx \frac{\log n}{k}$. That would unfortunately lead to each triangle staying in the graph with probability $\approx \frac{\log^3 n}{k^3}$ only, and in particular some vertices may end up participating in no triangles in the sample with high probability. The latter means that the corresponding singleton cuts in $G'_\Delta$ would be empty, and~\eqref{eq:gdgpd} would certainly not be satisfied.  Naturally, we can also try to sample each edge with a lower probability, say $\approx \frac{\log n}{k^{2/3}}$, but in this case the number of edges in the sparsifier of a $k$-regular graph would be $\approx k^{1/3}n \log n$, which is in superlinear in $n$.

In general, the above attempts point to the fact that edge weights in $G'_\Delta$ are a {\em non-linear function} of the random variables that govern the presence or absence of various edges in $G'$, making `one-shot' sparsification not easy to achieve. 

\paragraph{Essential problem of triangle-weighted graph.}

Although we have already outlined several problems that we encounter in our attempts to construct a sparsifier using the triangle-weighted graph, there is another fundamental problem which arises directly from its structure as the following example demonstrates.

Let the graph $G$ (see \cref{fig:dangling_triangle}) consist of a clique on vertices in $V(G) \setminus \{ v_1, \ldots v_l \}$ where $l=\lfloor \sqrt{n} \rfloor$, and let $h \leq \sqrt{n}/4$ be an integer. For $i \in [l]$, let vertex $v_i$ be connected with vertices $u_{i, 1}, \ldots, u_{i, l}$ such that the sets of neighboring vertices of $v_1, \ldots, v_l$ don't intersect. Notice that the subgraph induced by $V(G) \setminus \{ v_1, \ldots v_l \}$ has connectivity in $G_\Delta$ of at least $n(n - 1)/8$, forming a connectivity component, while vertices $v_1, \ldots, v_l$ are not a part of this component because they are only connected to the clique with at most $O(n)$ triangles each. 

In this situation, the triangles $v_i u_{i,j_1} u_{i, j_2}$ for $i \in [l]$, $j_1, j_2 \in [h]$, and $j_1 \neq j_2$ are ``dangling'', i.e. one of their edges is part of a component with a high connectivity, while there is no such component containing the whole triangle. 

We know from the first part of the introduction that there is a way to get a clique sparsifier with almost linear number of edges, and graph $G$ is a clique with additional $O(\sqrt{n})$ vertices and $O(n)$ edges. Since this is an insignificant part of the whole graph, one might think that it is still easy to get a sparsifier with almost linear number of edges, for example by taking all edges $v_i u_{i, j}$ with probability $1$, and sampling the clique as we did before. But we will now show that additional caution must be taken to handle the ``dangling" triangles.

\begin{figure}[h]
    \centering
    \includegraphics{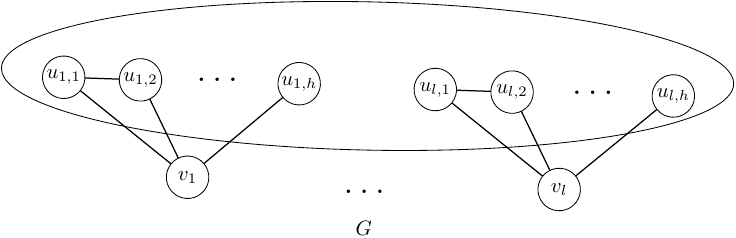}
    \caption{An illustration of ``dangling'' triangles.}
    \label{fig:dangling_triangle}
\end{figure}

First, suppose that we sample all of the edges in the clique with the probability $\leq 1/2$. Consider the case $h=2$. Then for all $i \in [l]$, edge $u_{i, 1}u_{i, 2}$ is contained in the only triangle in the cut $(\{ v_i \}, V(G) \setminus \{ v_i \})$. With high probability, at least one of those cuts will have motif size $0$ and therefore, the resulting graph would not be motif sparsifier.

On the other hand, suppose that we were to take all of the edges $u_{i, j_1} u_{i, j_2}$ with probability $1$. Consider the case of $h = \lfloor \sqrt{n}/4 \rfloor$. Then, the number of those edges would be $O(n^{3/2})$, and the sparsification would not produce any significant results. This shows that to take care of ``dangling'' triangles, we would need to sample the edges in them with different probabilities according to the situation at hand.

Under closer examination, one might discover that this problem stems from the following fact: consider a connectivity component $C$ in $G_\Delta$. If we were to take an induced subgraph of $G$ on vertices of $C$ and then build a triangle-weighted graph for it, the connectivity of this new triangle-weighted graph would most likely be lower than the connectivity of $C$.

As the above examples show, approaching motif sparsification purely from the point of view of sparsification of motif weighted graphs is difficult. Instead, we show, somewhat surprisingly, that a judicious composition of graph and hypergraph sparsification  methods leads to a very clean approach, which we describe next. After that, we demonstrate that motif weighted graph can still be used in the proposed framework to achieve a speed-up in running time for dense graphs.



\subsection{Strength-based sparsification}

As we have discussed, it seems that we can neither use hypergraph nor graph sparsification ideas directly to obtain motif sparsifiers. The reason for this is probably that a motif is an object that --- similarly to a hyperedge --- usually lives on sets of more than 2 vertices, but at the same time is composed of edges, i.e. it is closely related to graphs. As a consequence motif sparsification may be viewed as an intermediate problem between hypergraph and graph sparsification.

Indeed, our main contribution is to properly combine ideas from graph and hypergraph sparsification and to overcome some motif specific technical obstacles. Our starting point will be to extend the notion of strong connectivity that is an important ingredient to many sparsification approaches (see, for example, \cite{benczur_randomized_2015,chekuri2018minimum}) to the realm of motifs. Here we follow the hypergraph view and conceptually treat motifs as hyperedges. This way we can immediately extend the notion of connected components in hypergraphs \cite{chekuri2018minimum} to motifs by saying that a $k$-connected component is a maximal induced subgraph such that every cut is crossed by at least $k$ motif instances. This will allow us to define for each motif its importance as a measure of the amount it contributes to various cuts in the graph. 
The hypergraph view of motifs will also supply us with hypergraph cut counting arguments from \cite{kogan2015sketching} that can be easily transferred to motif cuts and that will be useful for the analysis. 

Once we have the definition of motif importance, it will be beneficial to switch to a graph-based view and think about how to compute the sparsifier. Our approach will be to sample edges but --- similarly to earlier work in graph sparsification --- we now need to identify important edges that we cannot miss for sparsification. In order to do so, we define the importance weight of an edge as the sum of the importance weights of its containing motifs. Edges whose importance weight is above a certain threshold will always be kept as sampling them would result in a variance that is too high.

For the remaining edges, we want to apply a sampling approach. Here, there are two more challenges. First, we need to deal with the non-linear behaviour of motif cut sizes and then we also need to address the fact that a motif is composed of several edges, which means that the events that two intersecting motifs are sampled is not independent, which means that we cannot use Chernoff bounds that are often used in the analysis of other sparsifying constructions. To deal with the non-linearity we observe that sparsifying by a constant factor is still possible and so we iteratively sparsify the graph $O(\log n)$ times by a constant. To deal with the dependencies in the sampling process and prove concentration, we use Azuma's inequality. During the different stages, edges that are no longer contained in any motif will receive a weight of $0$ and will then be dropped.

Finally, we observe that except for the sets of critical edges, all edges are sampled with the same probability and so we can use our approach to compute a sparsifier that works simultaneously for a set of motifs.

\subsection{Connectivity-based sparsification}
A major drawback of the strength-based algorithm is the need to enumerate all instances of a given motif. This task is hard, since enumeration takes time that is at least proportional to the number of motive instances, which in dense graph ($|E| = O(n^2)$) can easily reach $O(n^r)$.

However, the motif cut sparsification task doesn't implicitly require enumerating all of the motifs, and we show that by modifying an algorithm for exact subgraph counting  \cite{williams2013finding}, we can achieve sparsification in time $\tO(n^{\omega \lceil  r/ 3 \rceil })$, which is sublinear to the number of motifs in dense graphs.

The key idea is to move away from using the importances based on motif strengths to importances based on motif connectivities, where the connectivity of a motif instance is the minimal motif size of a cut crossing this instance. This new measure of importance can be approximated without needing to enumerate all motifs, which leads to the faster (in some settings) running time of our second algorithm.

In more detail, we adopt the sparsification approach of \cite{fung2019general} for our setting. A key object here is the motif weighted graph $G_M=(V, E, w_M)$, where, similarly to the triangle weighted graph, each edge $e$ is reweighted to $w_M(e)$ -- the sum of weights of motifs containing $e$. The main challenge is the approximation of motif-connectivity-based edge importance. This is done in two steps. First, we show that the connectivity of a motif instance is multiplicatively approximated by the minimum of motif connectivities of all edges in this instance, where the motif connectivity of an edge is the minimal size of a motif cut crossing this edge. Then, by dividing the graph into several layers, we are able to approximate the minimum-motif-connectivity-of-an-edge-based importance for each layer, which we then combine to get the final approximation.

The rest of the algorithm works in the same way as the first one, however we also use a result by \cite{abboud2021gomory} to compute all-pairs connectivities in $\tO(n^2)$ time. Our algorithm requires the motif connectivities of edges to be computed with multiplicative precision, which existing subquadratic approximation algorithms cannot deliver.

\subsection{Overview of Lower Bound}\label{sec:lower-bound-overview}

In \cref{sec:lower-bound}, we study the feasibility of producing a motif-cut sparsifier, similar to the one guaranteed by \cref{thm:intro-main}, in the setting where motif instances are required to be \emph{induced} subgraphs. The main difficulty in attempting to sparsify induced motifs is that the act of removing edges from $G$ may result in \emph{new} motif instances being created. This is not something we had to worry about in the proof of \cref{thm:intro-main}, and we could simply focus on preserving important motif instances that already existed in the original graph.

In fact, this difference turns out to result in a fundamental barrier, and we are able to show that any non-trivial sparsification may be impossible even for a motif as simple as the undirected $2$-path (see Theorem~\ref{thm:intro-graphlet-main}).

In our lower bound construction, the input graph will be the undirected, unweighted clique with the three edges of a specific triangle $(a,b,c)$ removed. More formally, we define $\Delta^-=(V, E_\Delta^-)$ as an unweighted, undirected graph on $n$ vertices, where
$$E_\Delta^-=\binom{V}{2}\setminus\big\{\{a,b\},\{b,c\},\{c,a\}\big\},$$
for distinct special vertices $a, b, c\in V$.

Note that while our Graph $\Delta^-$ is dense, the number of induced motifs is small, and each motif is of constant size. In the case of non-induced motif-sparsification, this setting would be trivial, as we could simply keep all edges contributing to any motif, thereby sparsifying the graph, and retaining the exact cut structure. In the case of induced motifs, however, this doesn't work, since removing edges may introduce additional motifs -- as it would in this example.

Specifically, the number of induced motif instances of the $2$-path motif is exactly $3(n-3)$, with each motif instance containing $2$ of the special vertices. In \cref{sec:graphlet-proof}, we essentially prove that any graph $\wh G$ that would sparsify $\Delta^-$ should have (at least some of) these same $2$-paths present. As it does in $\Delta^-$, this would result in a very large (quadratic) number of \emph{not-necessarily-induced} $2$-paths in $\wh G$. In order to insure that these aren't induced (and hence don't count as motif instances) $\wh G$ must be dense.

\begin{example}\label{ex:graphlet-sparsifier}
We give a slightly simpler -- but ultimately incorrect -- version of our above lower-bound construction for intuition. Consider the unweighted clique, with a single edge $(u,v)$ removed. More formally, $G=(V, E)$ is an unweighted undirected graph on $n$ vertices with
$$E=\binom{V}2\setminus\big\{\{u, v\}\big\}.$$

Attempting to sparsify this for the induced $2$-path motif, we can observe some of the same things as we do in our lower bound construction: Even though $G$ contains only a small, $\Theta(n)$, number of motifs, one cannot simply sparsify it by removing all edges that contribute to no motifs. The act of removing edges can create new induced $2$-paths, and we end up with a sparse graph whose induced-motif-cut structure doesn't resemble that of $G$ at all.

In fact, one can prove (in a similar manner to the proof of~\cref{thm:graphlet-main}) that no reweighted \emph{subgraph} of $G$ approximates its induced-motif-cut structure, for some small constant $\epsilon$. Surprisingly however, there does exist a weighted graph which achieves an arbitrarily close estimation: Let $\wh G=(V, \wh E, w)$ consist of the edge $\{u,v\}$ with weight $n^2$, and the $n-2$ edges in $u\times (V\setminus\{u,v\})$, each with weight $n^{-2}$. (This specificly gives a $(1\pm n^{-3})$-sparsifier, but the approximation can be arbitrarily improved by reweighting.) We leave the verification of the validity of this sparsifier to the reader.
\end{example}

%% file: prelims.tex
\section{Preliminaries}\label{sec:prelims}

Let $G=(V,E, w)$ be a directed weighted graph with vertex set $V=\{1,\dots,n\}$ and edge set $E \subseteq V \times V$, $m := |E|$. We will assume that the edge weights are always positive. Denote by $W = \max_{e \in E} w(e) / \min_{e \in E} w(e)$.
In this paper, we study the connectivity
structure of higher order patterns in the graph. More precisely, we consider a given directed graph $\M=(V_\M,E_\M)$ which we assume to be a
frequently occurring subgraph of $G$ and which we
refer to as a \emph{network motif} or \emph{motif} for short \cite{milo2002}. While the idea behind motifs
is that they are more frequently occurring than what one would expect in a random graph \cite{milo2002}, we are not making any formal assumption of this kind during the paper. Still, our motivation is that the motifs are common subgraphs.
We will always assume that motifs are weakly connected, i.e. the undirected version of the motif is connected. We make this assumption since we are interested in graph cuts; there is no convincing definition of a motif cut for motifs that have more than one weakly connected component.
Formally, we define motifs as follows.

\begin{definition}[Motifs and Motif Instances]
Let $ \M=(V_\M,E_\M)$ be a
weakly connected directed graph which we refer to as a \emph{motif}.
A subgraph of $G$ that is isomorphic 
to $\M$ is called an \emph{instance} of motif $M$ in $G$.
The set of all instances of a motif $M$ in $G$ is denoted
$\MG GM$.
\end{definition}

The definition of motifs extends to undirected graphs in a straightforward way by encoding undirected edges as two directed edges\footnote{Note that if the graph is weighted, the weight assigned to the two resulting edges should be equal to the square root of the weight of the original edge. This is because in Definition \ref{def:motif-weight} weight of a motif instance will be defined as the product of its edge weights.}.

We will be interested in weighted graphs and therefore require a definition of weights of motif instances. In order to obtain such a definition, we first consider integer weighted graphs.
A common interpretation of such graphs is that they can be viewed as unweighted multigraphs in which the multiplicity of each edge equals its weight. This 
view can be immediately generalized to define the weight of a motif of integer weighted graphs. We simply think of replacing every weighted edge by a corresponding number of copies and then count the number of distinct motifs. That is, the weight of a motif becomes the product of its edge weights.
 The extension to real non-negative weighted edges is straightforward. 

\begin{definition}[Weight of Motif Instance]\label{def:motif-weight}
Let $G=(V, E, w)$ be a directed weighted graph. The weight $w(I)$ of a motif instance $I=(V_I,E_I)$ in $G$ is defined as
$$
w(I) = \prod_{e\in E_I} w(e).
$$
\end{definition}

Let $(S, V \setminus S)$ be a cut in $G$. We say that motif instance $I=(V_I, E_I)$ crosses this cut if one of its edges crosses this \emph{undirected} cut. Since the motifs are weakly connected by definition, this is equivalent to $V_I \cap S \neq \emptyset$ and
$V_I \cap (S \setminus V) \neq \emptyset$.

\begin{definition}[Motif Size of a Cut]\label{def:motif-cut-size}
Let $G=(V,E, w)$ be a directed weighted graph. For a motif $M$ the \emph{$M$-motif size} of cut $(S, V \setminus V)$ is defined as
\[
    \Val_{M, G}(S, V \setminus S) = \sum_{I \in \mathcal M(G,M): I \text{ crosses } (S, V \setminus S)} w(I).
\]
\end{definition}
Note that the previous definition is directly influenced by applied works such as~\cite{benson2016, hao2017, tsourakakis2017} which also redefine the cut size in terms of the number of motifs crossing a cut.

Our goal is to construct an algorithm for sparsifying a graph in such a way that the motif sizes of all cuts are $(1 \pm \epsilon)$ preserved. We formalize this notion as follows.

\begin{definition}\label{def:motif-cut-sparsifier}
Let $\M=(V_\M, E_\M)$ be a motif and let $G=(V, E, w)$ be a directed weighted graph. A directed weighted graph $\widehat{G}$ is an $(M, \e)$-motif cut sparsifier of $G$, if for every cut  $(S, V \setminus S)$, the following holds:
\[
    (1 - \e) \Val_{M, G}(S, V \setminus S) \leq \Val_{M, \widehat{G}}(S, V \setminus S) \leq (1 + \epsilon) \Val_{M, G}(S, V \setminus S).
\]
\end{definition}

%% file: notations.tex
\subsection{Strong Motif Connectivity}

We now extend the notion of strong connectivity used in graph cut sparsification
\cite{benczur_randomized_2015} to motifs. 
For a given motif $M$ 
we will define the concepts of strong $M$-connectivity as well as $M$-connected components, which both follow naturally from the standard notion of strong connectivity.
Our notion is also closely related to strong connectivity in hypergraphs \cite{kogan2015sketching}, if we view a motif as a hyperedge. The main difference is that motifs are composed of simpler objects, i.e. edges. Similarly to the case of graphs and hypergraphs, strong motif connectivity will allow us to get bounds on the number of distinct cuts that we need to consider in the analysis.

In graphs and hypergraphs one can now define sampling probabilities for edges or hyperedges and sample them accordingly. These probabilities are based on a definition of the strength of edges.
It is tempting to follow the same approach for motif instances, however, as already discussed in the technical overview, there is a problem. If we sample a set of motif instances then their union may contain other motif instances that were not contained in the sample. The reason is simply that motifs are composed of edges.
Therefore, later on, we will define an edge-based sampling procedure. It will still be useful for our purposes to define a notion of motif strength. 
We now give the formal definitions.

\begin{definition}[Motif Connectivity]
Let $\M=(V_\M, E_\M)$ be a motif, let $G=(V, E, w)$ be a directed weighted graph.
$G$ is $(k, M)$-connected if 
every cut $(S,V\setminus S$) in $G$ has $M$-motif size at least $k$.
\end{definition}
\begin{definition}[$k$-Strong $M$-Connected Component]
\label{def:kconnectedcomp}
Let $\M=(V_\M, E_\M)$ be a motif, let $G=(V, E, w)$ be a directed weighted graph.
For a value $k\in \R_+$, an induced subgraph $C=(V_C,E_C, w)$ of $G$ is called a $k$-strongly $M$-connected component of $G$,
if 
\begin{enumerate}[label=(\alph*)]
    \item $C$ is $(k, M)$-connected and 
    \item there is no induced subgraph $C'=(V_{C'},E_{C'}, w)$ of $G$ that is $(k, M)$-connected and has $V_C \subsetneq V_{C'}$.
\end{enumerate}
\end{definition}

We will consider two $M$-connected components distinct if their sets of vertices are distinct.

\begin{definition}[Motif Strength]\label{def:motif_strength}
Let $\M=(V_\M, E_\M)$ be a motif, let $G=(V, E, w)$ be a directed weighted graph. Let $I \in \MG GM$ be a motif instance. The motif strength $\kappa_I$ of $I$ is the maximum value $k$ such that there exists a $(k, M)$-connected component that contains $I$ as a subgraph.
\end{definition}

%% file: algorithm.tex

\section{Main Results}\label{sec:main-results}

In this section we present the main results of this paper. We express runtime and size bounds in $\tO$ notation, which hides factors polynomial in $\log n$ and motif size. We start by stating the upper bound results in full generality:

\begin{restatable}{theorem}{maingeneral}\label{thm:intro-main-general}
Let $L>0$ be an integer.
For every directed weighted graph $G=(V, E, w)$, $|V|=n$, every set of motifs $\{\M_i\}_{i \in[L]}$
 and every $\e\in (0, 1)$, a graph $G'$ such that it is an $(M_i, \e)$-motif sparsifier of $G$ for all $i \in [L]$ with $\tO(L n/\epsilon^2)$ edges can be computed in time
 \[
 \tO\left(L|E| + \sum_{i=1}^{L} T(G, M_i) \right),
 \] 
where $T(G, M_i)$ for $i \in [L]$ is the time required to enumerate all instances of $M_i$ in $G$. The algorithm succeeds with probability at least $1 - (L + 1) n^{-c_1}$ for an arbitrarily large global constant $c_1$.
\end{restatable}

The main result of the paper is an immediate corollary:

\begin{corollary}\label{cor-main}
For every graph $G=(V, E)$, $|V|=n$, every constant integer $k\geq 2$, $\e\in (0, 1)$, there exists an $\e$-motif sparsifier $G'$ of $G$ with respect to all motifs $M$ of size at most $k$ simultaneously  that contains  $\tO(n/\e^2)$ edges. The graph $G'$ can be constructed in polynomial time.
\end{corollary}

The second algorithm provides the following guarantee:

\begin{restatable}{theorem}{maingeneral2}\label{thm:intro-main-general-2}
Let $L>0$ be an integer.
For every directed weighted graph $G=(V, E, w)$, $|V|=n$, every set of motifs $\{\M_i\}_{i \in[L]}$
 and every $\e\in (0, 1)$, a graph $G'$ such that it is an $(M_i, \e)$-motif sparsifier of $G$ for all $i \in [L]$ with $\tO(L n/\epsilon^2)$ edges can be computed in time
\[
    \tO(L (r^r + n^{\omega \lceil r / 3\rceil }) \log W),
\] 
where $r$ is the maximum number of vertices in $M_i$, $i \in [L]$, $W=\max_{e \in E} w(e) / \min_{e \in E} w(e)$ and $\omega < 2.37286$ is the matrix multiplication constant \cite{alman2021refined}.
The algorithm succeeds with probability at least $1 - (L + 1) n^{-c_1}$ for an arbitrarily large global constant $c_1$.
\end{restatable}

Although the two algorithms are very similar, there are cases when the first algorithm is faster than the second one. It would still be so even if we were to construct the motif weighted graph through enumeration. This is because computing all-pairs connectivities takes $\tO(n^2)$ time, while the first algorithm can work in nearly-linear time with respect to the number of motifs. This is relevant when, for example, we have only one motif --- triangle --- and $|E|=O(n^{4/3 - \delta})$ for $\delta > 0$. Then enumeration can be done in time $O(|E|^{3/2}) = O(n^{2 - 3 \delta /2})$ producing at most $O(|E|^{3/2})$ motif instances.

Last but not least, we derive a negative result on the possibility of constructing a motif cut sparsifier for {\bf induced} motif instances. The definitions of motif cut size and motif cut sparsifier are straightforwardly adapted from non-induced case by counting only induced motif instances as motif instances. See \Cref{sec:lower-bound} for details.

\begin{restatable}{theorem}{lowerbound}\label{thm:graphlet-main}
Let $f(n)=o(n^2)$ and let 
$\eps,0 < \eps \le 1/500$.
There exists a motif $M=(V_M,E_M)$  such that 
for every sufficiently large integer $n$, there exists a graph $G=(V, E)$ on $n$ vertices, such that it is impossible to construct an $(M,\epsilon)$-induced-motif cut sparsifier for $G$ with $f(n)$ non-negatively weighted edges.

\end{restatable}

Notice that this also includes graphs that are {\bf not} subgraphs of the original graph.

The paper is organised as follows. In \cref{sec:overview-parspar} we introduce \cref{alg:par_spars} (\parspar{}) for sparsifying an input graph by a constant factor while preserving its motif cut structure; we analyze this algorithm in \cref{sec:overview-parspar,sec:analysis-parspar}. Then, in \cref{analysis_motif} we introduce \cref{alg:motif_spars} (\textsc{MotifSparsification }), and prove that it achieves the guarantees of \cref{thm:intro-main-general}, which we prove at the end of the section.
Finally, in \cref{sec:lower-bound}, we present and prove our main lower bound result.

\section{Overview and analysis of \parspar{}}\label{sec:overview-parspar}

In this section we will develop the main algorithmic tool of this paper --- a procedure we call \parspar{} (\cref{alg:par_spars}) that with high probability sparsifies any graph (with sufficiently many edges) by a constant factor while approximately maintaining the motif cut sizes for a set of motifs. 
Once we have this procedure available, we can iterate it $\Theta(\log n)$ times to obtain our final sparsifier. Details can be found in \cref{analysis_motif}.

 Let $\{\M_1, \ldots, \M_L\}$ be a set of motifs. We aim to obtain a graph $G'$ such that it is a $( \e, M_i)$-motif cut sparsifier of $G$ simultanuously for \emph{all} $M_i$, $i \in [L]$. For $i\in[L]$, denote $r_i = |V_{\M_i}|$ and $r^*_i = |E_{\M_i}|$ as the size of the vertex and edge set of the $i$-th motif respectively; denote $r_{max} = \max_{i \in [L]} r_i$, $r^*_{max} = \max_{i \in [L]} r^*_i$ as the largest $r_i$ and $r^*_i$ value among all $i \in [L]$, respectively. As we will see, the running time and the sparsifier size depends on $r_{max}$ and $r^*_{max}$.

In our proofs, we will use a sufficiently small constant $d>0$, as well a constant $c_1>0$ which will govern the success probability of the algorithm. The value of $d$ depends on $c_1$; this dependency is determined in \cref{lem:layer_concentration}. The value of $c_1$ is arbitrary; we can for example assume that $c_1 = 10$ (this would ultimately lead to the failure probability being at most $n^{-4}$).


Our procedure \parspar{} is very simple. 
It identifies a set of \emph{critical edges} that have to be included in the sparsifier as sampling them would result in too high variance. The remaining edges will be taken with constant probability $p$. One may simply set $p=1/2$. However, our analysis implies that the size of the set of critical edges increases exponentially in $1/p$ and
it turns our that a better choice will be $p=2^{-1/{(2r^*_{max})}}$ as this balances the number of repetitions needed to sparsify the graph and the size of the set of critical edges.

We start by introducing definitions used in the algorithm.

\begin{definition}[Motif Weight of an Edge] \label{def:edge_motif_weight}
Let $\mathcal M=(V_\M,E_\M)$ be a motif and $G=(V,E, w)$ be a directed weighted graph.
Then the $\M$-motif weight of an edge $w_\M(e)$ is defined as \[w_\M(e) = \sum_{I \in \MG GM: e \in E(I)} w(I).\]
\end{definition}

\begin{definition}[Importance Weight]\label{def:edge_importance}
Let $\M=(V_\M,E_\M)$ be a motif and $G=(V,E, w)$ be a directed weighted graph.
Then 
\begin{itemize}
    \item for $I \in \MG GM$ the importance weight in $G$ is $\ws(I) = w(I) / \kappa_I$,
    \item for an edge $e \in E$, the $\M$-importance weight in $G$ is \[\ws_\M(e) = \sum_{I \in \MG GM: e \in E(I)} \ws(I).\]
\end{itemize}
\end{definition}

We now formally define our notion of critical edges.

\begin{definition}[Critical Edge]
\label{def:critical}
Let $\M=(V_\M,E_\M)$ be a motif and $G=(V,E, w)$ be a directed weighted graph.
An edge $e$ is called $\M$-critical if the $\M$-importance weight of $e$ is at least $\frac{d \epsilon'^2}{r^* (\log n + r)}$.
\end{definition}

While it is possible to compute the strengths of all motif instances exactly, it can be computationally expensive. Instead, we will approximate them.

\begin{lemma}[Follows from Theorem 6.1 of \cite{chekuri2018minimum}, Strength Estimation]
\label{lem:strength_est}
There exists algorithm \textsc{StrengthEstimation} which does the following:
it receives as an input a directed weighted graph $G=(V,E, w)$ and a motif instance set $\MG GM$ for a motif $\M=(V_\M,E_\M)$ and outputs strength estimations $\kappa'_I$ for each motif instance $I \in \MG GM$ with the following properties:
\begin{enumerate}
    \item For all $I \in \MG GM$, $\kappa'_I \leq \kappa_I$,
    \item $\sum_{I \in \MG GM} \frac{w(I)}{ \kappa'_I} \leq cr(n - 1)$, for some constant $c > 0$
\end{enumerate}
where $r = |V_M|$.
The running time of the algorithm is $O(r|\MG GM| \log^2 n \log (r |\MG GM|))$.
\end{lemma}

We defer the discussion of this algorithm and proof of this lemma to \cref{sec:str_est}. 

Since we don't have access to the motif instance strengths in our algorithm, we need to define a version of importance weight that uses strength approximations.

\begin{definition}
Let $\M=(V_\M,E_\M)$ be a motif and $G=(V,E, w)$ be a directed weighted graph.
Let $\kappa'_I$ be the estimations produced by the algorithm from \cref{lem:strength_est} for the graph $G$ for the motif $M$. 
\begin{itemize}
    \item For $I \in \MG GM$ the estimated importance weight is $\wse(I) = w(I) / \kappa'_I$,
    \item For an edge $e \in E$, the estimated $M$-importance weight is \[\wse_\M(e) = \sum_{I \in \MG GM: e \in E(I)} \wse(I).\]
\end{itemize}
\end{definition}

We can now present the \cref{alg:par_spars}.

\begin{algorithm}[H]
\caption{Partial Sparsification}
\label{alg:par_spars}
\begin{algorithmic}[1]
\Procedure{\textsc{PartialSparsification}}{$V, E, w, \epsilon', \{ \MG G{M_i} \}_{i=1}^{L} $} 
\State Calculate $\MG G{\M_i}$ for all $i \in [L]$ for graph $G=(V, E, w)$
\State $E_+ \gets \emptyset$
\For{$i = 1 \to L$}
\State $ \{ \kappa'_I\}_{I \in \MG{G}{M_i}} \gets \Call{StrengthEstimation}{G, \MG G{M_i}} $
\State $E_+ \gets E_+ \cup \{ e \in E: \wse_{\M_i}(e) \geq \frac{d \epsilon'^2}{ r^*_i (\log n + r_i) } \}$
\label{line:find_crit}
\Comment{Find critical edges}
\EndFor
\State $E_- \gets \emptyset$
\State $w'\gets w$
\For{$e\in E \setminus E_+$}
\If{a probability $p=2^{-1/{(2r^*_{max})}}$ Bernoulli variable is equal to $1$} 
\State $w'(e)\gets w(e)/p$
\Else
\State $w'(e)\gets 0$
\State $E_-\gets E_-\cup \{e\}$
\EndIf
\EndFor
\State $E \gets E \setminus E_-$
\State \Return $(E, w')$
\EndProcedure
\end{algorithmic}
\end{algorithm}

%% file: analysis_partial.tex
\subsection{Analysis of $E_+$} \label{sec:analysis-parspar}
In this and following subsections, we will show that \parspar{} indeed produces an $(\e', M)$-motif cut sparsifier of the input graph $G=(V,E)$.

Fix one of the motifs among $\M_1, \ldots, \M_L$ as $M$, and denote $r = |V(M)|$, $r^* = |E(M)|$. From here on we will show a number of properties of our algorithm that would eventually allow us to show that the final graph is a $M$-motif cut sparsifier. Since it will generally not involve any other motifs, we will omit mentioning $M$ in subscripts and other places where appropriate until \cref{subsec:multi_motif}.

$E_+$ is a set produced by $\parspar$. Recall that we want to sample each critical edge with probability $1$ and each other edge with probability $p$. In fact, what happens in the algorithm is that all of the edges in $E_+$ are sampled with probability $1$ and all other edges are sampled with probability $p$. Therefore, to show the correctness of the algorithm, it is necessary to show that $E_+$ contains all of the critical edges. Moreover, since on each iteration the graph loses about $1 - p$ fraction of the edges not in $E_+$, it is necessary for us to bound the number of edges in $E_+$ in order to bound the number of edges in the output graph. We do both in this section.

Denote
\[E_{M+} = \left\{ e \in E_1: \wse_{\M}(e) \geq \frac{d \epsilon'^2}{ r^* (\log n + r) } \right\}.\] 
First, we show that $E_{M+}$ (and consequently $E_+$) contains all $M$-critical edges:

\begin{lemma}
\label{lem:crit_correct}
The set $E_+$ in line \ref{line:find_crit} of \cref{alg:par_spars} contains all $M$-critical edges.
\end{lemma}
\begin{proof}
By \cref{lem:strength_est},
\[
    \forall I \in \MG GM: \kappa'_I \leq \kappa_I.
\]
By definition of $M$-critical edge and importance weight,
\[
    \sum_{I \in \MG GM: e \in E(I)} \frac{w(I)}{\kappa_I} \geq \frac{d \epsilon'^2}{r^* \log n}.
\]
On the other hand,
\[
    \sum_{I \in \MG GM: e \in E(I)} \frac{w_\M(I)}{\kappa_I} \leq \sum_{I \in \MG GM: e \in E(I)} \frac{w_\M(I)}{\kappa'_I} = \wse_\M(e).
\]
Therefore, the condition in line \ref{line:find_crit} of \cref{alg:par_spars} holds for $e$, which means that $E_+$ contains all $M$-critical edges.
\end{proof}

Furthermore, we have that $E_+ = \bigcup_{i=1}^{L} E_{M_i+}$, hence by bounding $E_{M+}$ we can bound $E_+$.

\begin{lemma}
\label{lem:crit_counting}
The size of $E_{M+}$ is at most 
$\frac{c r (r^*)^2 (n - 1) (\log n + r)}{d \epsilon'^2}$.
\end{lemma}
\begin{proof}
By \cref{lem:strength_est}, the following holds:
\[
    \sum_{I \in \MG GM} \wse(I) \leq c r (n - 1).
\]
We can bound the sum of estimations of importance weight of all edges:
\[
    \sum_{e \in E} \wse_\M(e) = \sum_{e \in E} \sum _{\substack{I \in \MG GM:\\ e \in E(I)}} \wse(I) \leq
    r^* \sum_{I \in \MG GM} \wse(I) \leq c r r^* (n - 1).
\]
On the other hand, since edges in $E_{+M}$ must satisfy inequality in line \ref{line:find_crit} of \cref{alg:par_spars}, we have
\[
    |E_{M+}| \cdot \frac{d \e'^2}{r^* (\log n + r)} \leq \sum_{e \in E_{M+}} \wse_\M(e) \leq \sum_{e \in E} \wse_\M(e).
\]
Combining both inequalities yields:
\[
    |E_{M+}| \leq \frac{c r (r^*)^2 (n - 1) (\log n + r)}{d \epsilon'^2},
\]
as desired.
\end{proof}

\begin{corollary}
\label{cor:total_crit_counting}
$E_+$ satisfies
\[|E_+| \le \frac{c L r_{max} {(r^*_{max})^2} (n - 1) (\log n + r_{max})}{d \epsilon'^2}. \]
\end{corollary}
\begin{proof}
The proof follows from \cref{lem:crit_counting} by summing across all motifs.
\end{proof}


\subsection{Correctness of \textsc{PartialSparsification}}

As in \cite{benczur_randomized_2015}, to show the correctness of our algorithm, we want to split our graph $G$ into a ``sum'' of several weighted graphs. The decomposition may be viewed as the motif-version of the decomposition of Benczur and Karger \cite{benczur_randomized_2015} and follows rather closely their ideas. 

Let $k_1, \ldots, k_h$ be all of the different strong $M$-connectivity values in $G$ in increasing order, i.e. for each $k_i$ there exists a $k_i$-connected component that is not $k$-connected for any $k>k_i$. Let $k_0 = 0$. 
In order to decompose our graph into a sum of weighted graphs we observe that we can write
\begin{eqnarray*}
\Val_{M, G}(S, V \setminus S) &=& \sum_{\substack{I \in \MG GM: \\ I \text{ crosses }(S, V \setminus S)}} w(I)\\
& = &
\sum_{\substack{I \in \MG GM: \\ I \text{ crosses } (S, V \setminus S)}} \frac{w(I)}{\kappa_I} \cdot \kappa_I \\
& = & \sum_{\substack{I \in \MG GM: \\ I \text{ crosses } (S, V \setminus S)}}
\frac{w(I)}{\kappa_I} \cdot \left(\sum_{i: k_i \le \kappa_I} k_i - k_{i-1}\right) \\
&=& \sum_{i=1}^h (k_i - k_{i-1}) \cdot 
\sum_{\substack{I \in \MG GM: \\
\kappa_I \ge k_i}} \frac{w(I)}{\kappa_I}
\cdot {\mathbbm 1}(I \text{ crosses } (S, V \setminus S))
\end{eqnarray*}
where the third equality follows from $\kappa_I = \sum_{i:\kappa_I \ge k_i} k_i-k_{i-1}$  
and where
${\mathbbm 1}(I \text{ crosses } (S, V \setminus S))$ denotes the indicator function that motif instance $I$ crosses the cut $(S,V\setminus S)$.
The above formula guides us towards our decomposition. The sum 
$$
\sum_{\substack{I \in \MG GM: \\
\kappa_I \ge k_i}} \frac{w(I)}{\kappa_I}\cdot {\mathbbm 1}(I \text{ crosses } (S, V \setminus S))
$$
ranges over all motif instances that are contained in components of $M$-connectivity at least $k_i$. We will now view the graph as a sum of graphs $F_i$ where each $F_i$ is the union of all $(k_i, M)$-connected components of $G$. The motif instances in the graph $F_i$ will be weighted by a factor of $k_i-k_{i-1}$.
In addition, each motif $I$ is reweighted by $1/\kappa_I$ in all graphs $F_i$. This motivates the following definition.



\begin{definition}
Let $\M=(V_\M,E_\M)$ be a motif and $G=(V,E, w')$ be a directed weighted graph. For a weighted graph $H=(V_H,E_H,w)$ with $V_H \subseteq V$ and $E_H \subseteq E$  and a cut $(S,V_H \setminus S)$
, we define the following value:
\[
    \IVal_{M, H}(S, V_H \setminus S) = \sum_{\substack{I \in \MG HM:\\
    I \text{ crosses } (S, V_H \setminus S)}} \frac{w(I)}{\kappa_I(G)},
\]
where $\kappa_I(G)$ is the motif strength of $I$ with respect to $G$.
\end{definition}
In the following we will always use the above definition in a way that $G$ is the input graph of \parspar{}. We will also frequently use $G$ as a subscript when the subgraph $H$ in the above definition equals $G$.
Using the above notation we can now write
\begin{eqnarray}
\Val_{M, G}(S, V \setminus S) 
& =& \sum_{i=1}^h (k_i-k_{i-1}) \cdot 
  \IVal_{M, F_i}(S, V \setminus S)\\
  &=& \sum_{i=1}^h (k_i-k_{i-1}) \cdot
  \sum_{\substack{\text{ $k_i$-connected}\\ \text{component } C}}  \IVal_{M, C}(S, V \setminus S)\label{eqn}
\end{eqnarray}
where the last equality splits $F_i$ into its  $M$-connected components.

Now, our goal is to show the concentration result for a single $M$-connected component.
There are two well-known results for hypergraph cuts that can be adapted for the case of motifs that we need to use to show concentration for all cuts. 
We include their proofs for completeness.

\begin{lemma}[Reformulation of Theorem 6.8 of \cite{chekuri2018minimum}]
\label{lem:weighted_cut_size}
Let $\M=(V_\M,E_\M)$ be a motif and $G=(V,E, w)$ be a directed weighted graph.
 If minimum of $
 \IVal_{M, G}$ of all non-trivial cuts is greater than $0$, it is equal to $1$.
\end{lemma}
\begin{proof}
Let $(S, V \setminus S)$ be the cut with minimum motif size and let $k$ be the $M$-connectivity of $G$. Then all of the motif instances crossing the cut have connectivity exactly $k$. On the other hand, the size of cut is $k$, which gives us
\[
\IVal_{M, G}(S, V \setminus S) = 
\sum_{I \in \MG GM: I \text{ crosses } (S, V \setminus S) } \frac{w(I)}{\kappa_I} = \frac{\Val_{M, G}(S, V \setminus S)}{k} = 1. 
\]
Consider any other cut $(S, V \setminus S)$ of motif size $k'$. Then the strength of all motif instances crossing this cut is at most $k'$, which means that:
\[
\IVal_{M, G}(S, V \setminus S) = \sum_{I \in \MG GM: I \text{ crosses } (S, V \setminus S) } \frac{w(I)}{\kappa_I} \geq \frac{\Val_{M, G}(S, V \setminus S)}{k'} = 1.
\]
Therefore, the cut with the minimum motif size is the cut with the minimum value of $\IVal_{M, G}$, and the latter is equal to $1$. 
\end{proof}

\begin{lemma}[Motif Cut Counting, Reformulation of Theorem 3.2 of \cite{kogan2015sketching}]
\label{cor:cut_counting}
Let $\M=(V_\M,E_\M)$ be a motif and $G=(V,E, w)$ be a directed weighted graph. Let $c$ be the minimum value of $\IVal_{M, G}$ across all cuts in $G$. There are at most $O(2^{
(\alpha + 1) r} n^{2(\alpha + 1)})$ cuts with value $\IVal_{M, G}$ at most $\alpha c$ for a real $\alpha \geq 1$.
\end{lemma}
The proof of the above lemma is given in \cref{sec:str_est}.

In the works on graph and hypergraph sparsification 
the value of a cut is defined by the edges or hyperedges. These are also the objects that are sampled and it suffices to use a Chernoff bound to analyze the concentration of a fixed cut. In our case, we sample edges but we are interested in the number of motif instances that cross the cut. We can write the cut value as a sum of random variables corresponding to the motifs that cross the cut, but these random variables are not independent and so we cannot use Chernoff bounds. To deal with dependencies we will instead use Azuma's inequality.

\begin{lemma}[Azuma's Inequality, \cite{azuma1967weighted}]
\label{lem:azuma}
Let $Z_0, Z_1, \ldots, Z_n$ be a martingale satisfying $|Z_i - Z_{i-1}| \le c_i$ for each $i \in [n]$. For any $\lambda > 0$,
\[ \p(|Z_n - Z_0| \ge \lambda) \le 2 \exp\left( \frac{-\lambda^2}{2(c_1^2 + \cdots + c_n^2)} \right). \]
\end{lemma}
We will use this lemma with a special `edge-exposure' martingale, which will be defined in the proof of \cref{lem:layer_concentration}.


\begin{lemma}
\label{lem:layer_concentration}
Let $C$ be a $M$-connected component before the application of \textsc{PartialSparsification} and $C'$ be that subgraph after the application. The following holds with probability at least $1 - n^{-c_1 }$:

Let $V_C = V(C)$. For all cuts $(S, V_C \setminus S)$ of $C$,
$$
    (1 - \epsilon') \IVal_{M, C}(S, V_C \setminus S) \leq \IVal_{M, C'}(S, V_C \setminus S) \leq (1 + \epsilon') \IVal_{M, C}(S, V_C \setminus S).
$$
\end{lemma}
\begin{proof}
By applying \cref{lem:weighted_cut_size} to $C$, we get that the minimum of values $\IVal_{M, C}$ across all cuts in $C$ is $1$.

Note that all of the critical edges in $C$ are critical in $G$, since their $M$-importance weight in $C$ is not larger than in $G$.
Since, by \cref{lem:crit_correct}, $E_+$ contains all critical edges in $G$, it also contains all of the critical edges in $C$, and, therefore, no critical edges are being sampled afterwards.
Fix a cut $(S,V_C\setminus S)$ with $\IVal_{M, C}(S,V_C\setminus S)=\alpha \geq 1$ in $C$. Let $E_2$, $t:=|E_2|$, be the set of all edges that are being sampled in \textsc{PartialSparsification} with probability $p=2^{-1/{(2r^*_{max})}} \geq 2^{-1/{(2r^*)}}$ and that are part of at least one motif cut by $(S, V_C \setminus S)$. Then $E_2$ does not contain any critical edges. Let $I_0$ be the subgraph of $C$ containing all vertices and edges that are a part of some motif  that is being cut by $(S, V_C \setminus S)$.

Consider the following random process: enumerate the edges in $E_2$ in the order they are being examined by \textsc{PartialSparsification}. Suppose edge $e$ with number $i$ is being sampled. If $e$ is not sampled, then we define $I_i$ to be equal to $I_{i-1}$ without $e$, otherwise $I_i$ is equal to $I_{i-1}$ with $e$ with it's weight multiplied by $1/p$. Now consider a random process $X_i$, $i \in \{0, \ldots, t\}$, where $X_i$ is equal to $\IVal_{M, I_i}(S, V_C \setminus S)$. 

It is easy to see that $X_t = \IVal_{M, C'}(S, V_C \setminus S)$ and that $X_i$ is a martingale. Let 
\[
\M^S(e) = \sum_{\substack{I\in\MG{C}{M}:\\ e \in E(I),\\ I \text{ crosses } (S, V_C \setminus S) }} \frac{w_\M(I)}{\kappa_I}.
\]
Then $|X_i - X_{i - 1}| \leq \M^S(e)/p^{r^*}$, where $e$ is the edge being sampled on step $i$,
since the weight of all motifs can change by at most $1/p^{r^*}$ during the random process. Note that although 
the weight of $e$ changes at most by a factor of $1/p$ the weights of other edges of any motif may have increased by a factor of $1/p$ earlier in the process. Since $\M^S(e)$
is defined at the beginning of the process, we can only bound
the increase by a factor of $1/p^{r^*}$.
Therefore, to apply \cref{lem:azuma}, we need to bound $\sum_{e \in E_2} \M^S(e)^2$.

Since there are no critical edges in $E_2$, for all edges $e$ that we sample, we must have 
\[\M^S(e) \leq \frac{d \epsilon'^2}{r^* (\log n + r)}.\] 
On the other hand,
\[
    \sum_{e \in E_2} \M^S(e) \leq r^* \sum_{I\in\MG{C}{M}: I \text{ crosses } (S, V_C \setminus S) } \frac{w_\M(I)}{\kappa_I} = r^* \alpha,
\]
since every motif contains at most $r^*$ edges. Therefore, using this inequality:
\[  
\sum_{e \in E_2} \M^S(e)^2 \le \frac{d \epsilon'^2 }{r^* (\log n + r)} \sum_{e \in E_2} \M^S(e) \le \frac{\alpha d \epsilon'^2 }{\log n + r}.
\]
Hence, because $\IVal_{M,C}(S, V_C \setminus S) = \alpha$ and by \cref{lem:azuma},
\begin{align*}
    &\p(|\IVal_{M,C}(S, V_C \setminus S) - \IVal_{M,C'}(S, V_C \setminus S)| \geq \epsilon' \IVal_{M,C}(S, V_C \setminus S)) \leq \p(|X_t - X_0| \ge \epsilon' \alpha)  \\
    &\le 2 \exp\left( \frac{-(\epsilon' \alpha )^2}{2\sum_{e \in E_2}\M^S(e) / p^{2r^*}} \right) \le 2 \exp\left( \frac{-(\epsilon' \alpha)^2}{4 \cdot \frac{\alpha d \epsilon'^2 }{\log n + r}} \right) = 2 \exp\left( \frac{-\alpha (\log n + r)}{4 d} \right).
\end{align*}

We now apply a union bound on all cuts $(S, V_C \setminus S)$ in conjunction with \cref{cor:cut_counting}. We need to bound $\sum_{\alpha \geq 1} P(\alpha)g(\alpha)$, where $P(\alpha)$ is the probability that the inequalities in the statement of the lemma doesn't hold for the cut with $\IVal_{M, C}$ equal to $\alpha$, $g(\alpha)$ is the number of those cuts, and the sum is taken across all values of $\alpha$ that are present in the graph.

Let $F(\alpha) = \sum_{\alpha \geq \alpha' \geq 1} g(\alpha')$ be the total number of cuts with $\IVal_{M, C} \leq \alpha$. By \cref{cor:cut_counting}, $F(\alpha) = A 2^{(\alpha + 1)r} n^{2(\alpha + 1)}$ for some constant $A$. We then adversarialy extend $F(\alpha)$ in a to the whole $\R_+$ such that $F(\alpha)$ is differentiable while preserving the above inequality.

We have that
\[
\sum_{\alpha \geq 1} P(\alpha)g(\alpha) \leq \int_{1}^\infty P(\alpha) \frac{dF(\alpha)}{d\alpha} d \alpha.
\]
Therefore, by applying partial integration,
\begin{align*}
&\int_{1}^\infty P(\alpha) \frac{dF(\alpha)}{d\alpha} d \alpha = \Big[ P(\alpha) F(\alpha)  \Big]_{1}^{\infty} - \int_{1}^\infty F(\alpha) \frac{dP(\alpha)}{d\alpha} d \alpha\\
&\leq 2 A 2^{2r} n^4 \exp \left(-\frac{\log n + r}{4d}\right) + 
\int_{1}^\infty A n^{2(x + 1)} 2^{r(x + 1)} \frac{\log n + r}{4d} 2\exp \left(-\frac{x (\log n + r)}{4d} \right) dx \\
&\leq n^{-c_1 }
\end{align*}
by setting $d$ to be sufficiently small. Thus, the inequalities in the lemma statement hold with probability at least $1 - n^{-c_1 }$ for all cuts, as desired.
\end{proof}

We will now use the fact that for a given cut, we can take the weighted sum of the $\IVal_{M, C}$ of cuts of each of the connectivity components $C$ such that this sum is equal to the motif cut size in the whole graph. We can then apply \cref{lem:layer_concentration} to each term to obtain the cut preservation property for the whole graph.

\begin{lemma}
\label{lem:cut_component_preservation}
Let $G'$ be $G$ after the application of \textsc{PartialSparsification}. $G'$ is $(M, \e')$-motif cut sparsifier of $G$ with probability $1 - n^{-c_1 + 3}$.

\end{lemma}
\begin{proof}
By definition of motif cut sparsifier, it is enough to show that the following holds for all cuts $(S, V \setminus S)$ of $G$:
\[
    (1 - \epsilon')\textup{Val}_{M, G}(S, V \setminus S) \leq \textup{Val}_{M, G'}(S, V \setminus S) \leq (1 + \epsilon') \textup{Val}_{M, G}(S, V \setminus S).
\]
By equation (\ref{eqn}),  we have 
\begin{eqnarray*}
\Val_{M, G}(S, V \setminus S) 
  &=& \sum_i (k_i-k_{i-1}) \cdot
  \sum_{\substack{\text{ $k_i$-connected}\\ \text{component } C}}  \IVal_{M, C}(S, V \setminus S).
\end{eqnarray*}
Now let $C'$ be $C$ after the application of \textsc{PartialSparsification}.
Then, similarly, the following holds:
\begin{eqnarray*}
\Val_{M, G'}(S, V \setminus S) 
 &=& \sum_i (k_i-k_{i-1}) \cdot 
\sum_{\substack{\text{ $k_i$-connected}\\ \text{component } C \text{ in } G}} \IVal_{M,C'}(S,V\setminus S).
\end{eqnarray*}
Note that if two $M$-connected components intersect, one of them is contained inside the other, and the smaller one has higher connectivity. Therefore, the set of all $M$-connected components is a laminar family, which means that its size is at most $2n$.
By applying \cref{lem:layer_concentration} to all $(k_i, M)$-connected components $C$ for all $i$ and a union bound over the at most $2n$ different $M$-connected components,  the following holds for all $M$-connected components:
\[
    (1 - \epsilon')\IVal_{M, C'}(S, S \setminus V) \leq \IVal_{M, C}(S, S \setminus V) \leq (1 + \epsilon') \IVal_{M, C'}(S, S \setminus V)
\]
with probability at least $1 -  2 n^{-c_1 + 1} \geq 1 -  n^{-c_1 + 3}$.
Combining all of the equalities and inequalities, we get the claim.
\end{proof}

\subsection{Hypergraphs}
We introduce hypergraphs here since we will use some results concerning them. A hypergraph is the pair of two sets $(V, F)$, where $V$ is the set of vertices and $F$ is the set of hyperedges $f$, which are subsets of $V$. Weighted hypergraph $H=(V, F, w)$ is a hypergraph with weight function $w: F \to \R_+$.
A hypergraph $H$ is $r$-uniform if every $f \in F$ satisfies $|f| = r$.
We denote the size of the cut $(S, V \setminus S)$ in hypergraph $H$ as $\textup{Val}_H(S, V \setminus S)$.

\begin{definition}[Induced Subhypergraph]
A hypergraph $H'=(V', F', w')$ is an induced subhypergraph of a hypergraph $H=(V, F, w)$, if $V' \subseteq V$, $F' = \{ f \in F: f \subseteq V' \}$ and $w$ and $w'$ are equal on $F'$.
\end{definition}

We will abuse the cut notation for the hypergraphs: if $H'=(V', F', w)$ is an induced subhypergraph of $H=(V, F, w)$, then $\textup{Val}_{H'}(S, V \setminus S) := \textup{Val}_{H'}(S \cap V', V' \setminus S)$.

\begin{definition}[Connectivity]
A weighted hypergraph $H=(V,F,w)$ is $k$-connected if 
every cut $(S,V\setminus S$), $S\neq \emptyset$, $S\subsetneq V$, in $H$ has size at least $k$.
\end{definition}
 
\begin{definition}[$k$-connected Component]
\label{def:hyper-kconnectedcomp}
For a weighted hypergraph $H=(V, F, w)$ with non-negative hyperedge weights and a value $k\in \R_+$, an induced subhypergraph $C=(V_C,F_C, w)$ of $H$ is called a $k$-connected component of $H$,
if 
\begin{enumerate}[label=(\alph*)]
    \item $C$ is $k$-connected and,
    \item there is no induced subhypergraph $C'=(V_{C'},F_{C'}, w)$ of $G$ that is $k$-connected and has $V_C \subsetneq V_{C'}$.
\end{enumerate}
\end{definition}

\begin{definition}[Hyperedge Strength]
Let $H=(V, F, w)$ be a weighted hypergraph with non-negative hyperedge weights. A hyperedge $f \in F$ has strength $\kappa_f$ if $\kappa_f$ is the
maximum value of $k$ such that there exists a $k$-connected component of $H$ that contains $f$.
\end{definition}

\subsection{Strength Estimation and Motif Cut Counting}
\label{sec:str_est}
To reiterate, construction of motif cut sparsifier is not possible by only using the techniques for constructing hypergraph sparsifier. But, the problems are sufficiently close to share some similarities, which allows us to use some results for hypergraph cut sparsification in our proof.

In this section we will present omitted proofs of \cref{lem:strength_est} and \cref{cor:cut_counting} by reducing them to similar existing results for hypergraphs.

Because motif instances are essentially just subsets of vertices, it is useful to consider them as hyperedges of some hypergraph, which we will call a motif hypergraph. Note that some motif instances share the same set of vertices. In this case, the weight of the resulting hyperedge is equal to the sum of their weights.

\begin{definition}[Motif Hypergraph]
Let $\M=(V_\M,E_\M)$ be a motif and $G=(V,E, w)$ be a directed weighted graph.
Then the $\M$-motif hypergraph of $G$ is an undirected weighted hypergraph $H_\M=(V, F_\M, w_\M)$, where
\begin{itemize}
    \item $F_M = \{ V(I): I \in \MG GM \}$,
    \item for $f \in F_M$, $w_M(f) = \sum_{I \in \MG G M: f = V(I)} w(I)$.
\end{itemize}
\end{definition}

Note that the motif hypergraph represents the motif connectivity structure of a graph: for a cut $(S, V \setminus S)$ in $G$, its motif size is equal to it's size in $H$, and for a $I \in \MG GM$, $\kappa_I = \kappa_{V(I)}$ where $V(I)$ is a hyperedge of motif hypergraph. The introduction of hypergraph allows us to use several results from hypergraph cut sparsification, as well as giving a new perspective on the problem.


We now prove \cref{lem:strength_est} by using the following result.

\begin{lemma}[Theorem 6.1 of \cite{chekuri2018minimum}, Strength Estimation]
\label{lem:strength_est_hyper}
There exists algorithm \textsc{StrengthEstimation} which does the following:
it receives as an input a rank $r$ weighted hypergraph $H=(V, F, w)$ on $n$ vertices and outputs strength estimations $\kappa'_f$ for each hyperedge $f$ with the following properties:
\begin{enumerate}
    \item For all $f \in F$, $\kappa'_f \leq \kappa_f$,
    \item $\sum_{f \in F} \frac{w(s)}{ \kappa'_f} \leq cr(n - 1)$, for some constant $c > 0$.
\end{enumerate}
The running time of the algorithm is $O(r|F| \log^2 n \log (r |F|))$.
\end{lemma}

Note that although the algorithm presented in \cite{chekuri2018minimum} can only work with natural weights, we can easily reduce the general case to it by dividing all of the weights by the minimum one and then rounding them down to the nearest integer: tt only worsens the second property by a factor of $2$.

\begin{proof}[Proof of \cref{lem:strength_est}]
We construct motif hypergraph $H_M$ and run the algorithm from \cref{lem:strength_est_hyper} on it, then set $\kappa'_I := \kappa'_{V(I)}$ for $I \in \MG GM$. Since the construction of $H_M$ takes only $O(|\MG GM|)$ time, the total runtime is the same as in \cref{lem:strength_est_hyper}, and both properties straightforwardly follow from properties of $H_M$.
\end{proof}

Finally, we give the proof of \cref{cor:cut_counting}.

\begin{lemma}[Cut Counting in Hypergraphs, Theorem 3.2 of \cite{kogan2015sketching}]
\label{lem:cut_counting_hyper}
In an $r$-uniform weighted hypergraph $H=(V, M, w)$ with size of minimum cut $c$, there are at most $O(2^{
\alpha r} n^{2\alpha})$ cuts of size no more than $\alpha c$ for a half-integer $\alpha \geq 1$ where $\alpha$ is a half-integer if $2 \alpha$ is an integer.
\end{lemma}

\begin{proof}[Proof of \cref{cor:cut_counting}]
Consider a modification of a motif hypergraph, where each hyperedge's weight is divided by its strength. Denote it by $H'$.
It is easy to see that the size of an arbitrary cut $(S, V \setminus S)$ in $H'$ is equal to the $\IVal_{M, G}(S, V \setminus S)$. 
Indeed, 
\begin{align*}
     \Val_{H'}(S, V \setminus S) &= \sum_{f \in F: f \text{ crosses } (S, V \setminus S)} \frac{w_M(f)}{\kappa_f} \\
    &=\sum_{I \in \MG GM: I \text{ crosses } (S, V \setminus S)} \frac{w(I)}{\kappa_I} \\
    &= \IVal_{M, G}(S, V \setminus S).
\end{align*}
Therefore, it is enough to show that if $c$ is the size of the smallest cut in $H'$, the number of cuts of size $\alpha c$ is at most $O(2^{(\alpha + 1)r}n^{2(\alpha + 1)})$, which we achieve as follows: find the smallest half integer $\beta \geq \alpha$ and apply \cref{lem:cut_counting_hyper} to it and $H'$. The result then follows from the fact that $\beta < 1 + \alpha$.
\end{proof}

\subsection{Runtime of \parspar{}}

\begin{theorem}
\label{thm:par_spar_time}
Let a directed weighted graph $G=(V, E, w)$, $\e' \in (0, 1)$ and a set of motifs $\{\M_i\}_{i \in[L]}$ be the input of \parspar. The total running time of \parspar is 
\[
\sum_{i=1}^{L} T(G, M_i) + \tO\left(L|E| + \sum_{i=1}^{L}|\MG G {\M_i}|\right),
\] 
where $T(G, M_i)$ for $i \in [L]$ is the time required to enumerate all instances of $M_i$ in $G$.
\end{theorem}
\begin{proof}
We will analyze each of the procedures. \textsc{StrengthEstimation} takes time
$O(r_i|\MG{G}{M_i}| \cdot\log^2 n\cdot \log (r_i|\MG{G}{M_i}|))$ by \cref{lem:strength_est} for each $M_i$. Computing the values $\wse_{\M_i}(e)$ and finding all critical edges
can be done in $O(r^*_i |\MG G {\M_i}| + |E|)$ time for $i \in [L]$. We repeat those steps for all $L$ motifs. Sampling edges in the loop requires $O(|E|)$ operations. On top of that, the algorithm calculates $\MG G {\M_i}$ and $H_{M_i}$, which requires $\sum_{i=1}^{L} T(G, M_i) + O(\sum_{i=1}^{L} {r^*_i} |\MG G {\M_i}|)$ time, resulting in the total running time of 
\[
\sum_{i=1}^{L} T(G, M_i) + O\left(\sum_{i=1}^{L} (r_i|\MG{G}{M_i}|\cdot \log^2 n\cdot \log (r_i|\MG{G}{M_i}|) + r^*_i |\MG G {\M_i}| + |E|)\right). \qedhere
\]
\end{proof}

%% file: analysis_motif.tex

\section{Analysis of \textsc{MotifSparsification}}
\label{analysis_motif}

We are now ready to analyze the complete algorithm, \motspar{}. As was mentioned before, it essentially only calls \parspar \, $O(r^*_{max} \log n)$ times.
Hence, our main goal in this section is to show that after all these applications, the graph is still $(\e, M_i)$-motif sparsifier for all $i \in [L]$.

Because we also have the \Cref{alg:fast_par_spars} utilizing the same sparsification approach, we will show a proof for a generic algorithm, \gparspar{}, which abstracts both of the partial sparsification algorithms.

\begin{definition}
We assume that \gparspar{} accepts as input a weighted directed graph $G=(V, E, w)$, $\eps' > 0$,  and a set of motifs $ \{ M_i \}_{i=1}^{L}$, and returns $(E', w')$ such that $G'=(V, E', w')$ is $(\eps', M_i)$-motif cut sparsifier for all $i \in [L]$ with probability at least $1 - n^{-c_1}$ obtained through sampling at most $A$ of the edges with probability $1$ and the rest of the edges with probability $2^{-1/(2r^*_{max})}$ in time $B$. $A$ and $B$ can depend both on input parameters, as well as on the constant $c_1$.
\end{definition}

\begin{algorithm}[H]
\caption{Motif Sparsification}
\label{alg:motif_spars}
\begin{algorithmic}[1]
\Procedure{\textsc{MotifSparsification}}{$G, \{\M_i\}_{i \in[L]}, \epsilon$}
\Comment{$G=(V, E, w)$}
\State $\e' \gets \frac{\epsilon}{5 c_1 r^*_{max} \log n }$ \Comment{$c_1$ is an absolute constant, $r^*_{max}$ is maximum motif size}
\State $E_0\gets \emptyset$
\For{$j=1$ to $\lceil 2 c_1 r^*_{max} \log n \rceil $} \label{line:inner_loop}
\State $(E, w)\gets \gparspar{}(V, E, w, \epsilon', \{ M_i \}_{i=1}^{L})$
\EndFor
\State \Return $G = (V, E, w)$
\EndProcedure
\end{algorithmic}
\end{algorithm}

Since the approximation error grows multiplicatively after each application of \gparspar{}, we will need \cref{lem:exp_bounds} to get a final approximation bound.

\begin{lemma}
\label{lem:final_size}
At the end of the \textsc{MotifSparsification}, the set $E$ contains at most 
\(
A \)
edges with probability at least $1 - n^{-c_1 + 2}$.
\end{lemma}
\begin{proof}
Since all edges, except for those that are sampled with probability $1$, are sampled independently with same probability (and we only care about their quantity) and since the number of edges sampled with probability $1$ is bounded by $A$, we can assume without loss of generality that those are the same edges each time. 

Consider an arbitrary edge $e \in E$ at the start of the loop. Assume that $e$ is present in $E$ at the end of the algorithm. If it was sampled each time with probability $2^{-1/{(2r^*_{max})}}$, the probability of this happening is at most \[(2^{-1/{(2r^*_{max})}})^{2 c_1 r^*_{max} \log n} = n^{-c_1}.\] 
Since there are at most $n(n - 1)$ edges in $G$, the probability that at least one of those edges will be present in $E$ is less than $n^{-c_1 + 2}$ by a union bound. Therefore, $E$ consists entirely of edges sampled with probability $1$ with probability at least $1 - n^{-c_1 + 2}$, hence it's size is at most $A$.
\end{proof}

\begin{lemma}
\label{lem:single_motif_spars}
Let a directed weighted graph $G=(V, E, w)$, $\e \in (0, 1)$ and a set of motifs $\{\M_i\}_{i \in[L]}$ be the input of \textsc{MotifSparsification} and $G'$ be its output. Then for an arbitrary $M \in \{\M_i\}_{i \in[L]}$, $G'$ is $(M, \e)$-motif cut sparsifier of $G$ with probability at least $1 - n^{-c_1 + 5}$ for a sufficiently large $n$.
\end{lemma}
\begin{proof}
By definition, $G'$ is a $(M, \e)$-motif cut sparsifier if for all cuts $(S, V \setminus S)$, we have
\[
    (1 - \e) \Val_{M, G}(S, V \setminus S) \leq \Val_{M, G'}(S, V \setminus S) \leq (1 + \e) \Val_{M, G}(S, V \setminus S).
\]
We now proceed to show that the above inequalities hold. Denote by $l = \lceil 2c_1 r^*_{max} \log n \rceil$ the number of loop iterations and
denote by $G_j=(V, E_j, w_j)$ the state of the graph at the end of the loop iteration $j$ where $G_0 = G$. We will prove the following inductive statement:

With probability at least $1 - j n^{-c_1 + 3}$, the following holds after loop iteration $j$:
 For any cut $(S, V \setminus S)$ of $G$, the following holds:
\[
    (1 - \epsilon')^{j} \textup{Val}_{M, G}(S, V \setminus S) \leq \textup{Val}_{M, G_j}(S, V \setminus S) \leq (1 + \epsilon')^{j} \textup{Val}_{M, G}(S, V \setminus S).
\]
Note that for $j \leq l$, $(1 + \epsilon')^{j} \leq (1 + \frac{\epsilon}{2l})^{l} \leq 1 + \epsilon \leq 2$ and, similarly $(1 - \epsilon')^{j} \geq 1 - \epsilon/2 \geq 1/2 $ by \cref{lem:exp_bounds} and since $\epsilon \leq 1$.

\noindent \emph{Base case}: for $j = 0$, the property is trivial.

\noindent \emph{Inductive step}: suppose that the statements hold for $j - 1$. We can apply \cref{lem:cut_component_preservation}, which, combined with inductive assumption, gives us the property.

The probability that the used lemma fails is at most $n^{-c_1 + 3}$. Therefore, by a union bound with the probability that inductive assumption holds, the probability that the statement for iteration $j$ holds is at least $1 - j n^{-c_1 + 3}$.

Since $G' = G_{l}$, the inductive assumption on the last iteration also holds for $G'$, which means that:
\[
    (1 - \e)\textup{Val}_{M, G}(S, V \setminus S) \leq (1 - \epsilon')^{l} \textup{Val}_{M, G}(S, V \setminus S) \leq \textup{Val}_{M, G'}(S, V \setminus S),
\]
which gives us the desired lower bound. The upper bound is proven similarly. In total, the failure probability is at most \( l \cdot n^{-c_1 + 3} \), which is less then \( n^{-c_1 + 5} \) for a sufficiently large $n$.
\end{proof}

\subsection{Multiple Motifs}
\label{subsec:multi_motif}
We now put together all of our preceding lemmas to get our final sparsification result for all motifs simultaneously. 

\begin{lemma}
\label{thm:final_spars}
Let a directed weighted graph $G=(V, E, w)$, $\e \in (0, 1)$ and a set of motifs $\{\M_i\}_{i \in[L]}$ be the input of \textsc{MotifSparsification} and $G'$ be it's output. Then with probability at least $1 - L \cdot n^{-c_1 + 5}$, $G'$ is $\M_i$-motif cut $(1 + \epsilon)$ sparsifier of $G$ for all $i \in [L]$ for a sufficiently large $n$.
\end{lemma}
\begin{proof}
The proof follows from applying \cref{lem:single_motif_spars} to each of the motifs $M_i$, $i \in [L]$.
\end{proof}

\begin{lemma}
\label{thm:final_time}
Let a directed weighted graph $G=(V, E, w)$, $\e \in (0, 1)$ and a set of motifs $\{\M_i\}_{i \in[L]}$ be the input of \textsc{MotifSparsification}. The total running time of \textsc{MotifSparsification} is $O(r^*_{max} B \log n)$.
\end{lemma}
\begin{proof}
Immediate from the fact that $B$ is the runtime of $\gparspar{}$.
\end{proof}

\subsection{\motspar{} with \parspar}

\begin{proof}[Proof of \cref{thm:intro-main-general}]
The proof follows from \cref{thm:final_spars}, \cref{lem:final_size} and \cref{thm:final_time}.

By \Cref{cor:total_crit_counting}, the number of edges in the final graph is at most
\begin{align*}
    \frac{c L r_{max} {(r^*_{max})^2} (n - 1) (\log n + r_{max})}{d \epsilon'^2} &= O \left(\frac{L r_{max} {(r^*_{max})^4} (n - 1) (\log n + r_{max}) \log^2 n}{\epsilon^2} \right)\\ &= \tO(Ln / \e^2).
\end{align*}
To improve upon runtime a little bit, in \parspar, we can compute each set of motif instances only once, since we only delete them during the algorithm. Since the algorithm calls \textsc{PartialSparsification} in a loop the total runtime is 
\begin{align*}
    &\sum_{i=1}^{L} T(G, M_i) + \\
    &O\left(\sum_{i=1}^{L} {(r^*_{max}
   )}^2 |\MG G {\M_i}|\cdot \log n + r_{max}^* r_{i}|\MG{G}{M_i}| \cdot\log^3 n \cdot\log (r_{i}|\MG{G}{M_i}|) + |E|) \right) \\
   &= \sum_{i=1}^{L} T(G, M_i) + \tO\left(L|E| + \sum_{i=1}^{L}|\MG G {\M_i}|\right). \qedhere
\end{align*}
\end{proof}

%% file: sublinear_intro.tex
\section{Sparsification without enumeration}

One of the main problems of the presented algorithm is that it requires finding every motif instance, which takes time at least equal to the number of motif instances, which can reach $O(n^r)$ in dense graphs.

Nevertheless, there is still a way to circumvent the enumeration. Consider the case when the sparsification is performed with respect to only one motif $M$. Recall that \textsc{PartialSparsification} on a high level does two things: finds critical edges and samples non-critical edges with high probability. The importance of the edge is defined as the sum of importances of motifs containing this edge:
\[
    \ws_\M(e) = \sum_{I \in \MG GM: e \in E(I)} \ws(I),
\] where $\ws(I) = w(I) / \kappa_I$,
and the edge is critical if $\ws_\M(e) \geq \frac{d \epsilon'^2}{r^* (\log n + r)}$.

It is easy to see that all steps of this procedure can be performed in time $\tO(|E| + |V|)$, except for computing the values $\ws_\M(e)$.
This is why we opt for a different approach of defining importances, based on connectivities.

\subsection{Basic Definitions}

\begin{definition}[Motif Connectivity]\label{def:weak_motif_strength}
Let $\M=(V_\M, E_\M)$ be a motif, let $G=(V, E, w)$ be a directed weighted graph. Let $I \in \MG GM$ be a motif instance. The connectivity, $k_I$ of $I$ is the minimum $M$-motif size of a cut $(S, V \setminus S)$ which $I$ crosses.
\end{definition}

Accordingly, we adopt the following notation. 

\begin{definition}[Connectivity Importance Weight]

Let $M = (V_M, E_M)$ be a motif and $G=(V,E,w)$ be a directed weighted graph. Then
\begin{itemize}
    \item for $I \in \MG GM$, the connectivity importance weight in $G$ is $\wws(I) = w(I) / k_I$,
    \item for an edge $e \in E$, the connectivity $M$-importance weight in $G$ is
    \[
    \wws_\M(e) = \sum_{I \in \MG GM: e \in E(I)} \wws(I).
    \]
\end{itemize}
\end{definition}

While it is unclear how to compute motif strengths without enumerating all motifs, there is a way to approximate motif connectivities. The key idea is to compute the motif weighted graph, and then use the edge connectivities there to bound the motif connectivities, since the cut sizes in motif weighted graph are close to the $M$-motif sizes of corresponding cuts in the original graph.
\begin{definition}[Motif Weighted Graph]\label{def:motif-weighted-graph}
Let $\M=(V_\M, E_\M)$ be a motif and $G=(V, E, w)$ be a directed weighted graph. The undirected graph $G_M = (V, E, w_M)$ is called the $M$-motif weighted graph. (Recall from \cref{def:edge_motif_weight} that $w_\M(e) = \sum_{I \in \MG GM: e \in E(I)} w(I)$.) The motif weighted graph should be considered as undirected.
\end{definition}

Although it was shown by \cite{fung2019general} that graph cut sparsification is possible using the importances based on connectivities, to our knowledge no previous work has shown that it is possible in the hypergraph setting. Hence to show the correctness of the proposed algorithm, we shall adapt their techniques to our approach.

Finally, to compute the motif weighted graph we shall modify an algorithm for computing the number of motif instances in the graph \cite{williams2013finding}.

The following sections will be organized as follows: we will first present the algorithm for computing the motif weighted graph and prove its correctness and runtime, followed by the sparsification algorithm. In the rest of this section we will introduce necessary definitions and show some of their properties.

\begin{definition}[$M$-connectivity of an edge]
Let $M = (V_M, E_M)$ be a motif and $G=(V,E,w)$ be a directed weighted graph. Recall that the connectivity of an edge $e$ in graph $G_M$ is the minimum size of a cut cutting $e$ in $G_M$. For $e \in E$, the value $k_{M, e}$, equal to the connectivity of an edge $e$ in the motif weighted graph $G_M$, is called $M$-connectivity of the edge $e$.
\end{definition}

We will be omitting subscript $M$ where possible.

\begin{definition}

Let $M = (V_M, E_M)$ be a motif and $G=(V,E,w)$ be a directed weighted graph. Then
\begin{itemize}
    \item for $I \in \MG GM$, the estimated connectivity importance weight in $G$ is 
    \[
    \wwse(I) = w(I) \frac{r^*}{\min_{e \in E(I)} k_{M, e}},
    \]
    \item for an edge $e \in E$, the estimated connectivity $M$-importance weight in $G$ is
    \[
    \wwse_\M(e) = \sum_{I \in \MG GM: e \in E(I)} \wwse(I).
    \]
\end{itemize}
\end{definition}

\begin{lemma} \label{lem:motif_conn_approx}
Let $M = (V_M, E_M)$ be a motif and $G=(V,E,w)$ be a directed weighted graph. For any cut $(S, V \setminus S)$, the following holds:
\[
    \Val_{M, G}(S, V \setminus S) \leq \Val_{G_M}(S, V \setminus S) \leq r^* \Val_{M, G}(S, V \setminus S).
\]
In addition, for all $I \in \MG GM$,
\[
    \wws(I) \leq \wwse(I) \leq r^*\cdot \wws(I).
\]
\end{lemma}

\begin{proof}
The first property follows from the following observation: if a cut cuts a motif, then it cuts between $1$ and $r^*$ of its edges. Therefore, the contribution of a motif $I$ to the size of a cut that crosses it in $G_M$ is between $w(I)$ and $r^* w(I)$.

To establish the second property, notice that for all $I \in \MG GM$, by definition of edge connectivity,
\[
    \min_{e \in E(I)} k_{M, e} = 
    \min_{e \in E(I)} \min_{\emptyset \subsetneq S \subsetneq V: (S, V\setminus S) \text{ cuts } e} \Val_{G_M}(S, V\setminus S) = 
    \min_{\emptyset \subsetneq S \subsetneq V: (S, V\setminus S) \text{ cuts } I} \Val_{G_M}(S, V\setminus S).
\] 
This, in combination with the first property, leads to
\[
    k_I \leq \min_{e \in E(I)} k_{M, e} \leq r^* k_I,
\]
which implies the second property.
\end{proof}

%% file: motif_weighted_graph.tex
\subsection{Constructing the motif weighted graph}
Most of the algorithmic ideas in this section are adopted from \cite{williams2013finding}. The design of the algorithm is based on the idea of reducing the task of computing the number of motif instances in graph $G$ to the task of computing the number of triangles in a specially constructed graph $G_\sigma$, with a one-to-one correspondence between motif instances in $G$ and triangle instances in $G_\sigma$. Then, we can apply fast matrix multiplication to count the number of triangles in $G_\sigma$.

The problem of computing a motif weighted graph is slightly different from the problem of computing the number of motifs. We can use the latter in a black box manner: for each edge $e \in E$, we delete this edge from the graph and compute the number of remaining motif instances. The difference between this number and the number of motif instances in the original graph is the motif weight of the edge $e$. This approach requires calling the motif counting primitive $|E| + 1$ times. In this subsection, we present an algorithm which can construct the motif graph without this additional factor of $|E|$ in the running time.

\subsubsection{Notation}

Most of the notation in this subsection is exclusive to this subsection. We fix the motif $M=(V_M, E_M)$, and we will be omitting it where possible. Denote by $T_k = \{ (v_{1}, \ldots, v_{k}): \forall j \in [k] \: v_{j} \in V \land \forall j, i \in [k] \: j \neq i \rightarrow v_i \neq v_j \}$ the set of all ordered sequences of $k$ distinct vertices of $G$.

Let $k_1, k_2, k_3$ be such that $k_1 + k_2 + k_3 = r$, and $ \lfloor r/3 \rfloor \leq k_i  \leq \lceil r/3 \rceil$. The algorithm starts by constructing a tripartite graph with weighted vertices and edges 
$G_{\sigma} = (V_\sigma, E_\sigma, w_\sigma)$ defined as follows: $V_\sigma = T_{k_1} \cup T_{k_2} \cup T_{k_3}$, where we consider the entries of the three parts $T_{k_i}$ to be distinct. Fix some arbitrary ordering $\pi$ of vertices $V_M$, and let $\pi_1$ be its first $k_1$ entries, $\pi_2$ the next $k_2$ entries, and $\pi_3$ the rest of its entries. 

For a vertex $v \in T_{k_i} \subset V_\sigma$, consider a natural mapping $f_{v}: v \to V_M$,
where for $\ell \in [k_i]$, $f_v(v_l) = \pi_{i, \ell}$. Note that if we are also given $u \in T_{k_{j}}$ and $z \in T_{k_{h}}$, where $i, j, h$ are $1, 2, 3$ in some order, and $v, u, z$ are pairwise disjoint, we can construct natural extensions of the corresponding mappings: $f_{v,u} = f_v \cup f_u$ and $f_{v, u, z} = f_v \cup f_u \cup f_z$. Notice that both of them, as well as $f_v$, are bijections. We call such a mapping $f$ consistent if $f^{-1}$ is a graph homomorphism (that is every edge in $M$, when mapped via $f^{-1}$, corresponds to an edge in $G$). We denote by $E(f)$ the subset of edges $E$ that are mapped to edges in $E_M$.

For a vertex $v \in T_{k_i} \subset V_\sigma$, its weight is defined to be equal to
\[
    w_\sigma(v) = \prod_{e \in E(f_v)} w(e),
\] if $f_v$ is consistent, and $0$ if it is not. For a pair of vertices $u, v \in V_\sigma$, there is an edge between them if they come from different sets $T_{k_i}$, are pairwise disjoint, and mapping $f_{u, v}$ is consistent. The weight of the edge $(u, v)$ is equal to
\[
    w_\sigma((u, v)) = \prod_{e \in E(u, v)} w(e),
\] where $E(u, v) = \{ e \in E(f_{u, v}): |e \cap u| = |e \cap v| = 1\}$.

\begin{lemma} \label{lem:mwg_bijection}
    $P$ is an injective graph homomorphism between $G_M$ and a subgraph $G'$ of $G$ iff there exists a \bf{triangle} $u, v, z \in V_\sigma$ such that $P=f^{-1}_{u, v, z}$.
\end{lemma}
\begin{proof}
The reverse direction is easy to see. Since there is only an edge between two vertices if they are pairwise disjoint, come from different sets $T_{k_i}$, 
and $f_{u,v},
f_{v,z}$, and $f_{v,z}$ are consistent, it follows
that $f_{u, v, z}$
is consistent and, therefore, $P$ is a homomorphism.

In the other direction, suppose that $P$ maps
$\pi_1$ to $u$, $\pi_2$ to $v$, $\pi_3$ to $z$. Since $P$ is injective, $u$, $v$ and $z$ are pairwise disjoint and don't contain repeating elements. Therefore, they are vertices of $G_\sigma$, and, since $P$ is homomorphism, by definition they are pairwise connected by edges. By definition of $f_{u, v, z}$, $f^{-1}_{u, v, z} = P$.
\end{proof}

The approach now is to compute the triangle weighted graph for $G_\sigma$, and use it to construct the motif weighted graph of the original graph. The triangle weighted graph we construct differs somewhat from the motif weighted graph defined in \cref{def:motif-weighted-graph}, since we must take into account the vertex-weights in $G_\sigma$. Formally, denote by $\Delta$ the triangle motif, i.e. the clique on $3$ vertices. For $I \in \MG {G_\sigma} \Delta$, define
\[
    w_\sigma(I) = \prod_{v \in V(I)} w_\sigma(v) \prod_{e' \in E(I)} w_\sigma(e').
\]
For $e \in E_\sigma$, define
\[
w_{\Delta, \sigma}(e) = \sum_{I \in \MG {G_\sigma} \Delta: e \in E(I)} w_\sigma(I)
\]
and for $v \in V_\sigma$,
\[
w_{\Delta, \sigma}(v) = \sum_{I \in \MG {G_\sigma} \Delta: v \in V(I)} w_\sigma(I).
\]
Let $A$ denote the number of automorphisms of $M$.

\begin{lemma} \label{lem:mwg_counting}
    For $e \in E$, 
    \[
        w_M(e) = \frac{1}{A} \left( \sum_{v \in V_\sigma: e \in E(f_v)} w_{\Delta, \sigma}(v) + \sum_{(u, v) \in E_\sigma: e \in E(u, v)} w_{\Delta, \sigma}((u,v)) \right)
    \]
\end{lemma}
\begin{proof}
Let $w'_M(e)$ be the weighted sum of all vertex-ordered instances of $M$ containing $e$. Then, trivially, $w'_M(e) = A \cdot w_M(e)$. Each vertex-ordered instance of $M$ is uniquely defined by an injective homomorphism $P$ from $G_M$ to a subgraph of $G$, with its weight being:
\[
    w(P) = \prod_{e \in P(E_M)} w(e),
\] where $P(E_M) = \{(P(u), P(v)): (u, v) \in E_M \}$ is the projection of the edges of $M$.

By \cref{lem:mwg_bijection}, $P$ uniquely maps to a triple of vertices $u, v, z$ forming a triangle. Since the set of edges $P(E_M) = E(f_{u, v, z})$ can be partitioned into sets of edges between elements of $u$, $v$ and $z$, and between pairs of elements from different vertices,
\[
    w(P) = \prod_{e \in P(E_M)} w(e) = \prod_{v \in V(I)} w_\sigma(v) \prod_{e' \in E(I)} w_\sigma(e') = w_\sigma(I),
\] where $I = (u, v, z)$ is the aforementioned triangle. Therefore
\[
    w'_M(e) = \sum_{I \in \MG {G_\sigma} \Delta: e \in E(I) }  w_\sigma(I).
\]
Note that for a triangle $I=(u, v, z)$, edge $e \in E$ can only be present in one of the sets 
\[
E(f_v), E(f_u), E(f_z), E(u, v), E(u, z), E(v, z),
\] since $u, v, z$ are pairwise disjoint, and that those sets form a partition of $E(I)$. Therefore
\[
    w'_M(e) = \sum_{I \in \MG {G_\sigma} \Delta }  w_\sigma(I) = \sum_{v \in V_\sigma: e \in E(f_v)} w_{\Delta, \sigma}(v) + \sum_{(u, v) \in E_\sigma: e \in E(u, v)} w_{\Delta, \sigma}((u,v)).
\]
The claim now follows from the relation  $w'_M(e) = A \cdot w_M(e)$.
\end{proof}

\subsubsection{Analysis of the algorithm} We can now present the \Cref{alg:count_motifs}. Let $N(v)$ be the set of vertices adjacent to $v \in V_\sigma$.

\begin{algorithm}[H]
\caption{Constructing the motif weighted graph}
\label{alg:count_motifs}
\begin{algorithmic}[1]
\Procedure{\textsc{MotifWeights}}{$G=(V, E, w), M=(V_M, E_M)$} 
\State Compute number of automorphisms $A$ of $M$.
\State Construct graph $G_\sigma=(V_\sigma, E_\sigma)$. 
\State Let $W$ be the weighted adjacency matrix of $G_\sigma$.
\State Let $D$ be a diagonal matrix with diagonal entries $w_\sigma(v), v \in V_\sigma$.
\State $U \gets D W D W D$

\For{$(u, v) \in E_\sigma$}
    \State $w_{\Delta, \sigma}((u, v)) \gets W_{u, v} \cdot U_{u, v}$.
\EndFor
\For{$v \in V_\sigma$}
    \State $w_{\Delta, \sigma}(v) \gets \frac{1}{2} \sum_{u \in N(v)} w_{\Delta, \sigma}((u, v))$
\EndFor

\State $\forall e \in E: w'_M(e) \gets 0$
\For{$v \in V_\sigma$, $e \in E(f_v)$}
    \State $w'_M(e) \gets w'_M(e) + w_{\Delta, \sigma}(v)$
\EndFor
\For{$(u, v) \in E_\sigma$, $e \in E(u, v)$}
    \State $w'_M(e) \gets w'_M(e) + w_{\Delta, \sigma}((u, v))$
\EndFor
\State \Return $w'_M / A$
\EndProcedure
\end{algorithmic}
\end{algorithm}

\begin{theorem} \label{lem:motif_counting_correcntess}
Let a directed weighted graph $G=(V, E, w)$ and a motif $M=(V_M, E_M)$ be the input of \textsc{MotifWeights}. Then \cref{alg:count_motifs} returns the function $w_M(e), e \in E$.
\end{theorem}
\begin{proof}
Assuming that values $w_{\Delta, \sigma}$ are computing correctly by the algorithm, \Cref{lem:mwg_counting} implies that the values $w'_M$ and $w_M$ are also computed correctly. Therefore, we only need to prove correctness of computation of $w_{\Delta, \sigma}$.

To show that, notice that $w_\sigma((u,v)) = W(u, v) = 0$ if $u$ and $v$ are not connected. Therefore, for $(u, v) \in E_\sigma$:
\begin{align*}
    w_{\Delta, \sigma}((u,v)) =& \sum_{I \in \MG {G_\sigma} \Delta: e \in E(I)} w_\sigma(I) \\
    &= \sum_{z \in V_\sigma: (u, z) \in E_\sigma \land (z, v) \in E_\sigma} w_\sigma(u) w_\sigma((u, z)) w_\sigma(z) w_\sigma((z, v)) w_\sigma(v) w_\sigma((u, v)) \\
    &= \sum_{z \in V_\sigma} D_{u,u} W_{u, z} D_{z, z} W_{z, v} D_{v, v} W_{u, v} = (DWDWD)_{u, v}W_{u, v},
\end{align*}
which is exactly what is being computed. Considering values $w_{\Delta, \sigma}$ for vertices, the following equality holds
\[
    w_{\Delta, \sigma}(v) = \frac{1}{2} \sum_{u \in N(v)} w_{\Delta, \sigma}((u, v))
\]
since each triangle containing $v$ will be counted twice in the sum on the right hand side.
\end{proof}

\begin{theorem} \label{thm:motif_counting}
Let a directed weighted graph $G=(V, E, w)$ and a motif $M=(V_M, E_M)$ be the input of \textsc{MotifWeights}. Then its running time is 
$O(n^{\omega \lceil r / 3 \rceil} + r^2 n^{2 \lceil r / 3 \rceil} + r^* r^r)$ where $n^\omega$ is matrix multiplication time.
\end{theorem}
\begin{proof}
The number of automorphisms can be computed in time $O((r^* + r) r^r)$, by checking all $r!$ permutations of vertices of $M$.

The graph $G_\sigma$ has $O(n^{\lceil r / 3 \rceil})$ vertices and $O(n^{2 \lceil r / 3 \rceil})$ edges, and can be constructed in time 
$O((r + r^*)n^{2 \lceil r / 3 \rceil})$. Using fast matrix multiplication
(\cite{alman2021refined} is state-of-the art at the time of writing), $U$ can be computed in time $O(n^{\omega \lceil r / 3 \rceil})$.

The rest of the algorithm can be computed in time $O(r^2 (|V_\sigma| + |E_\sigma|)) = O(r^2 n^{2 \lceil r / 3 \rceil})$, which gives the final runtime
\[
    O(n^{\omega \lceil r / 3 \rceil} + r^2 n^{2 \lceil r / 3 \rceil} + r^* r^r). \qedhere
\]
\end{proof}

%% file: fast_partial_sparsification.tex
\subsection{Fast Partial Sparsification}

In this subsection we will present the main part of the sublinear algorithm, \textsc{FastPartialSparsification}, which is a counterpart to \Cref{alg:par_spars}. It differs in the way it computes the edge importances. It uses \textsc{MotifWeights} algorithm, as well as an almost quadratic time all-pairs max-flow algorithm \cite{abboud2021gomory} \cite{abboud2021subcubic} to compute the motif weighted graph and motif edge connectivities. Then, using this information, the algorithm produces edge importance estimates and sample edges according to them.

We denote the ratio between the highest and the lowest weight by $W$.

The all-pairs max-flow problem is equivalent to the problem of computing edge connectivities between any two pairs of vertices. As was mentioned, we will utilize a result on its computation:
\begin{theorem}[Theorem 1.3 of \cite{abboud2021gomory}] \label{thm:all_pairs_conn}
For an undirected weighted graph $G=(V, E, w)$, $n = |V|$, there is a randomized Monte Carlo algorithm \textsc{Connectivities} for computing edge connectivities between all pairs of vertices that runs in time $\tO(n^2)$.
\end{theorem}

Recall that we are trying to approximate
$$\wwse_\M(e) = \sum_{I \in \MG GM: e \in E(I)} \wwse(I)=\sum_{I \in \MG GM: e \in E(I)}  w(I) \frac{r^*}{\min_{e \in E(I)} k_{M, e}}$$
for each edge $e\in E$. Armed with \textsc{MotifWeights} (\cref{alg:count_motifs}) and \textsc{Connectivities} (\cref{thm:all_pairs_conn}) we are able to calculate the values of $k_{M,e}$. However, calculating the above formula naively would still require us to sum over all motif instances, which is prohibitively slow.

Instead we split the graph in to levels based on the motif-connectivities of its edges as follows: Let $k_{min} = \min_{e \in E} k_e$. For $j \in \mathbb{N} \cup \{0\}$, let $G_j=(V, E_j, w)$, where $E_j = \{ e \in E: k_e \geq 2^j k_{min} \}$. Let $w_{M, j}(e)$ be the motif weight of edge $e \in G_j$. For $ I \in \MG GM$ denote $\rho_I = \min_{e \in E(I)} k_{M, e} $.
Notice that for $j \in \mathbb{N} \cup \{0 \}$, 
\[
\{ I \in \MG GM: 2^{j + 1} k_{min} > \rho_I \geq 2^{j} k_{min} \} = \MG{G_{j}}{M} \setminus \MG{G_{j + 1}}{M}.
\]

Instead of directly computing $\wwse_M$, we will use its approximation function $\wwsee_M: E \to \mathbb{R}$, where 
\[
    \wwsee_M(e) = \sum_{j=0}^\infty \quad \sum_{I \in \MG{G_{j}}{M} \setminus \MG{G_{j + 1}}{M}: e \in E(I)} w(I) \frac{r^*}{2^j}.
\]

As we'll show below, this quanity $\wwsee_M(e)$ is faster to calculate, yet approximates $\wwse$ sufficiently well that we can use it in the construction of our sparsifier.

\begin{lemma} \label{lem:wwse_approx}
Let $M = (V_M, E_M)$ be a motif and $G=(V,E,w)$ be a directed weighted graph. Then for $e \in E$,
\[
\wwse_M(e) \leq \wwsee_M(e) \leq 2 \wwse_M(e).
\]
\end{lemma}
\begin{proof}
For $I \in \MG{G_{j}}{M} \setminus \MG{G_{j + 1}}{M} $, $2^{j} \leq \rho_I < 2^{j + 1} $. Therefore
\begin{align*}
    \wwse_M(e) & = \sum_{I \in \MG{G}{M}: e \in E(I)} w(I) \frac{r^*}{\rho_I} 
    = \sum_{j=0}^\infty \sum_{I \in \MG{G_{j}}{M} \setminus \MG{G_{j + 1}}{M}: e \in E(I)} w(I) \frac{r^*}{\rho_I} \\
    & \leq \sum_{j=0}^\infty \sum_{I \in \MG{G_{j}}{M} \setminus \MG{G_{j + 1}}{M}: e \in E(I)} w(I) \frac{r^*}{2^j} = \wwsee_M(e).
\end{align*}
The upper bound can be shown similarly.
\end{proof}

On the other hand, $\wwsee$ can be easily computed using \Cref{alg:count_motifs}.

\begin{lemma}
Let $M = (V_M, E_M)$ be a motif and $G=(V,E,w)$ be a directed weighted graph. Let $\Lambda_M = \lceil \log \max_{e \in E} k_{M, e} / k_{min} \rceil$. Then for $e \in E$,
\[
    \wwsee_M(e) = \sum_{j=0}^{\Lambda_M} (w_{M, j}(e) - w_{M, j + 1}(e)) \frac{r^*}{2^j}.
\]
\end{lemma}
\begin{proof}
The lemma follows from the fact that
\[
    w_{M, j}(e) = \sum_{I \in \MG{G_{j}}{M}: e \in E(I)} w(I). \qedhere
\]
\end{proof}

We are ready to present the fast partial sparsification algorithm.

\begin{algorithm}[H]
\caption{Fast Partial Sparsification}
\label{alg:fast_par_spars}
\begin{algorithmic}[1]
\Procedure{\textsc{FastPartialSparsification}}{$V, E, w, \epsilon', \{ M_i \}_{i=1}^{L} $}
\State Rescale $w$ so that $\min_{e \in E} w(e) = 1$.
\State $E_+ \gets \emptyset$
\For{$i = 1 \to L$} \label{line:fast_spar_loop}
\State $ w_{M_i} \gets \Call{MotifWeights}{G=(V, E, w), M_i} $
\State $ \{ k_e\}_{e \in E} \gets \Call{Connectivities}{G_{M_i}=(V, E, w_{M_i})} $
\State  $k_{min} \gets \min_{e \in E} k_e$
\State  $E_j \gets \{ e \in E: k_e \geq 2^j k_{min} \}$
\For{$j = 0 \to \Lambda_{M_i}$}
\State $w_{M_i, j} \gets \Call{MotifWeights}{G_j=(V, E_j, w), M_i}$
\EndFor
\State Let
\[
    \Upsilon' \gets \frac{{\eps'}^2 }{256 (d_1 + r_i + 2r_i^*) (r_i^*)^2 r_i \log n \ln n}
\]
\State \label{line:add_crit_weak} $E_+ \gets E_+ \cup \{ e \in E: \wwsee_{\M_i}(e) \geq \Upsilon' \}$ 
\EndFor
\State $E_- \gets \emptyset$
\State $w'\gets w$
\For{$e\in E \setminus E_+$}
\If{a probability $p=2^{-1/{(2r^*_{max})}}$ Bernoulli variable is equal to $1$} 
\State $w'(e)\gets w(e)/p$
\Else
\State $w'(e)\gets 0$
\State $E_-\gets E_-\cup \{e\}$
\EndIf
\EndFor
\State $E \gets E \setminus E_-$
\State Rescale $w'$ with respect to original weights.
\State \Return $(E, w')$
\EndProcedure
\end{algorithmic}
\end{algorithm}

\subsection{Correctness of \textsc{FastPartialSparsification}}

The goal of this subsection is to show that the output of \textsc{FastPartialSparsification} is indeed a motif cut sparsifier of the original graph. The analysis closely follows that of \cite{fung2019general}, while accommodating for the fact that in our application we are dealing with a different sampling scheme. This proof can be adapted to show the possibility of cut sparsification in hypergraphs using connectivities, albeit, using our tools, the guarantees on the sparsifier size in this case is most likely not tight.

Notice that due to the way we are scaling the weights, it holds that $\min_{I \in \MG GM} k_I \geq 1$.

\begin{definition}
    Let $M = (V_M, E_M)$ be a motif and $G=(V,E,w)$ be a directed weighted graph. An instance $I \in \MG GM$ is called $k$-heavy if its connectivity $k_I$ is at least $k$. Otherwise, it is $k$-light.
\end{definition}

\begin{definition}
    Let $M = (V_M, E_M)$ be a motif and $G=(V,E,w)$ be a directed weighted graph. For a cut $(S, V \setminus S)$, its (motif) $k$-projection is the set of $k$-heavy motif instances crossing this cut. 
\end{definition}

While the concept of $k$-projection was originally conceived for regular cut sizes \cite{fung2019general}, our definition is more general as they are equivalent when the motif is just one edge. Hence we will refer to them as edge $k$-projections.

\begin{theorem}[Theorem 2.3 of \cite{fung2019general}] \label{thm:graph_proj_num}
Let $G=(V,E,w)$ be a weighted graph. Let $\lambda$ be the minimum size of a cut in $G$, $k \geq \lambda$ and $\alpha > 0$. Then the number of distinct edge $k$-projections in cuts of size at most $\alpha k$ is at most $n^{2 \alpha}$.
\end{theorem}

We now use Theorem \ref{thm:graph_proj_num} to prove the following lemma which extends the
edge case to arbitrary motifs.
The reader might notice that while the bound provided by \Cref{thm:graph_proj_num} matches that of cut-counting theorem of Karger \cite{karger1999random}, \Cref{lem:proj_num} does not match \Cref{lem:cut_counting_hyper}. While it would be desirable to match this bound if one were to construct a connectivity-based hypergraph cut sparsifier, this statement is sufficient for our application.

\begin{lemma} \label{lem:proj_num}
Let $M = (V_M, E_M)$ be a motif and $G=(V,E,w)$ be a directed weighted graph. Let $\lambda$ be the minimum motif size of a cut in $G$, $k \geq \lambda$ and $\alpha > 0$. Then the number of distinct $k$-projections in cuts of motif size at most $\alpha k$ is at most $n^{2 \alpha r^*}$.
\end{lemma}

\begin{proof}
We are going to show that the number of motif $k$-projections in cuts of motif size at most $\alpha k$ is at most the number of edge $k$-projections in cuts of size at most $\alpha r^* k$ in the motif weighted graph $G_M$, which is $n^{2\alpha r^*}$.

Indeed, consider a motif $k$-projection $P$ of some cut $(S, V \setminus S)$ of motif size $\alpha k$. Let $f$ be the following mapping: $f(P) = \bigcup_{I \in P} \{ e \in E(I): (S, V \setminus S) \text{ cuts } e\}$, and let $F = \bigcup_{I \in P}  E(I)$. 
Then, by \Cref{lem:motif_conn_approx}, the size of the cut $(S, V \setminus S)$ in $G_M$ is at most $\alpha r^* k$, and the sets $f(P)$, $F$ contain only $k$-heavy edges.
Let $\mathcal{P}$ be the set of edge $k$-projections in cuts of size at most $\alpha r^* k$. Let $\mathcal{P}' = \{ P \cap F: P \in \mathcal{P}\}$. Trivially, $|\mathcal{P}'| \leq |\mathcal{P}|$.
Now, notice that $f(P) \in \mathcal{P}'$, since each edge in $f(P)$ must be $k$-heavy, and the cut $(S, V \setminus S)$ induces an edge $k$-projection.

On the other hand, suppose that there are two motif $k$-projections $P_1$ and $P_2$ such that $f(P_1) = f(P_2)$. Let $(S, V \setminus S)$ be the cut inducing $P_1$. For any $I \in P_2$, there must be an edge $e \in E(I) \cap f(P_2)$. But then $(S, V \setminus S)$ cuts $e$ and, therefore, $I$. Therefore, $I \in P_1$, and $P_2 \subseteq P_1$. By exchanging $P_1$ and $P_2$ we also get that $P_1 \subseteq P_2$ and $P_2 = P_1$. Therefore,
$f$ is injective. 
But since $f$ maps each motif $k$-projection to $\mathcal{P}'$, their number is at most $|\mathcal{P}'| \leq |\mathcal{P}|$.
\end{proof}

Equipped with this result, we can show that it is possible to do partial sparsification if we only sample edges $e$ with high value of $\wws_M(e)$ with probability $1$. The idea is to divide each cut into parts containing edges with approximately the same connectivity and show the concentration of each part, which motivates the following definition.

\begin{definition} \label{def:overlap_graphs}
    Let $M = (V_M, E_M)$ be a motif and $G=(V,E,w)$ be a directed weighted graph. Let $\Lambda = \lceil \log \max_{I \in \MG GM} k_I \rceil$ and Let $F_i = \{I \in \MG GM: 2^i \leq k_I < 2^{i + 1} \}$ for $i \in [\Lambda]$.
    We define $H_i=(V, J_i, w_i)$ as follows: $J_i = \{ I \in \MG GM: w(I) \geq 2^{i - 1} / n^r \}$, $w_i(I) = \min (2^{i + 1}, w(I))$. Notice that $F_i \subseteq J_i$ and $\forall I \in F_i{:}\ w(I) = w_i(I)$. Let $\pi_i$
    be the minimum of the connectivities of the motifs in $F_i$ in the graph $H_i$.
\end{definition}
Notice that if the graph is unweighted, all of $H_i$ are just equal to the motif hypergraph $H$. We could have defined them to all be equal in all cases; then however, the second bound in the following lemma would have depended on logarithm of the ratio between the maximum and minimum weights in the graph.

\begin{lemma} \label{lem:conn_overlap}
Let $M = (V_M, E_M)$ be a motif and $G=(V,E,w)$ be a directed weighted graph. Let $\gamma = 2 r \log n$.
\begin{itemize}
    \item For all $i \in [\Lambda]$, $\pi_i \geq 2^{i - 1}$,
    \item For all cuts $(S, V \setminus S)$, \[
    \sum_{i = 0}^{\Lambda} \Val_{H_i}(S, V \setminus S) \leq \gamma \Val_{G, M}(S, V \setminus S).
    \]
\end{itemize}
\end{lemma}
\begin{proof}
We first show the first point. Consider any $I \in F_i$ and consider any cut $(S, V \setminus S)$ cutting $I$. If $(S, V \setminus S)$ cuts any other $I' \in \MG GM$ such that $w(I') \geq 2^{i + 1}$, then $\Val_{H_i}(S, V \setminus S) \geq 2^{i + 1}$. Otherwise, because $|\MG GM| \leq n^r$, there is at most $n^r$ hyperedges that were crossing the cut in $H_M$ that are not present in $J_i$. Since the sum of their weight is at most $2^{i - 1}$, $\Val_{H_i}(S, V \setminus S) \geq 2^{ i- 1}$ because $\Val_{G, M}(S, V \setminus S) \geq 2^i$ since $k_I \geq 2^i$. Therefore, each cut cutting $I$ has size at least $2^{i - 1}$, and $\pi_i \geq 2^{i - 1}$. 

Now, to show the second point, fix a cut $(S, V \setminus S)$.  For simplicity, let $w_i(I) = 0$ if $I \not\in J_i$, and let $j_I = \min \{ j \in \mathbb{N}: 2^{j + 1} \geq w_i(I) \} $. Notice that for any $I \in \MG GM$,
\[
    \sum_{i = 0}^\Lambda w_i(I) \leq \sum_{i=0}^{j_I} 2^{i + 1} + \sum_{i=j_I + 1}^{\Lambda} w(I) \leq  2^{j_I + 2} + (\log(n^r) + 1) w(I)
\] because $I$ is not present in $J_i$ for $i > \lfloor \log(n^r) \rfloor + j_I + 1$.

Then, since $2^{j_I} \leq w(I)$,
\begin{align*}
    \sum_{i = 0}^{\Lambda} \Val_{H_i}(S, V \setminus S) &= \sum_{I \in \MG GM: I \text{ crosses } (S, V \setminus S)} \sum_{i = 0}^\Lambda w_i(I) \\
    &\leq \sum_{I \in \MG GM: I \text{ crosses } (S, V \setminus S)} \left( 2^{j_I + 2} + (\log(n^r) + 1) w(I) \right)
    \\
    &\leq \sum_{I \in \MG GM: I \text{ crosses } (S, V \setminus S)} 2 r \log n \cdot w(I)
    \leq 2 r \log n \cdot \Val_{G, M}(S, V \setminus S)
\end{align*}
for a large enough $n$.
\end{proof}

Recall that for $e \in E$,
\[
    \wws_M(e) = \sum_{I \in \MG GM: e \in E(I)} \frac{w(I)}{k_I}.
\]

The overall strategy is for each $F_i$, we show for all cuts simultaneously that the difference between the true contribution of edges in $F_i$ to the size of this cut versus the one observed in the sparsified graph is bounded by $\eps$ times the size of this cut in graph $H_i$. Because we have a bound on the sum of the sizes of cuts in $H_i$ from the previous lemma, we can obtain a guarantee in terms of the size of this cut in the original graph.

\begin{theorem} \label{thm:iterative_conn_sampling}
Let $M = (V_M, E_M)$ be a motif and $G=(V,E,w)$ be a directed weighted graph. Let
\[
    \Upsilon = \frac{\eps^2 }{64 (d_1 + r + 2r^*) r^* \gamma \ln n}.
\]
 Consider the following sampling scheme: all of the edges $e \in E$ such that $\wws_M(e) \geq \Upsilon$ are sampled with probability $1$, and all other edges $e$ are independently sampled with probability $p_e \geq 2^{-1/(2r^*)}$, with their weights multiplied by $1/p_e$ after successful sampling. Then with probability at least $1 - 4 n^{-d_1 + r}$ for a global constant $d_1$, the graph $G'$ obtained after the sampling is a $(M, \eps)$-motif cut sparsifier of $G$.
\end{theorem}
\begin{proof}
Recall that we need to show that for all cuts $(S, V \setminus S)$,
\[
    |\Val_{G', M}(S, V \setminus S) - \Val_{G, M}(S, V \setminus S)| \leq \eps \Val_{G, M}(S, V \setminus S).
\]
Fix a cut $(S, V \setminus S)$. Let $F_{i, S} = \{I \in F_i: I \text{ crosses } (S, V \setminus S)\}$. Let $f_{i, S} = \sum_{I \in F_{i, S}} w(I)$, and $e_{i, S} = \Val_{H_i}(S, V \setminus S)$. Let $G'=(V, E', w')$ be the graph obtained after the sampling and let $H'_M$ be its motive hypergraph. Denote $f'_{i,S} = \sum_{I \in F_{i, S}} w'(I)$.

We start with the following lemma:

\begin{lemma} \label{lem:conn_bound_core}
For any fixed $i$ with probability at least $1 - 1/n^{d_1}$ all cuts $(S, V \setminus S)$ satisfy
\[
    |f_{i, S} -  f'_{i, S}| \leq \frac{\epsilon}{2} \max \left( \frac{e_{i,S} 2^{i - 1}}{\gamma \pi_i}, f_{i, S} \right) \leq \frac{\epsilon}{2} \max \left( \frac{e_{i,S}}{\gamma}, f_{i, S} \right).
\]
\end{lemma}
\begin{proof}
By \Cref{lem:conn_overlap}, $\frac{2^{i - 1}}{\pi_i} \leq 1$, which implies the second inequality.

For the first inequality, notice that if $f_{i, S} = 0$ then $F_{i, S}$ is empty and the lemma statement is trivially true with probability $1$ for the cut $(S, V \setminus S)$. Hence we can assume that $f_{i, S} > 0$ and that there is at least one motif instance in $F_{i, S}$. Since it must be at least $\pi_i$-connected in $H_i$, $e_{i, S} \geq \pi_i$. This means that we can split the remaining cuts into sets of the following form:
\[
    C_{i, j} = \{ (S, V\setminus S): \pi_i \cdot 2^{j} \leq e_{i, S} < \pi_i \cdot 2^{j + 1} \}
\] for $j \in \mathbb{N} \cup \{ 0 \}$.

We will show that with probability at least $1 - 2n^{-d_1 2^j}$, all of the cuts in $C_{i, j}$ satisfy the property. By the union bound, we will then have that the probability that any cut violates the property is at most
\[
    \sum_{j = 0}^{\infty} 2n^{-d_1 2^j} \leq 4 n^{-d_1}
\]
and we are done.

Now, fix $j \in C_{i, j}$ and a cut $(S, V \setminus S)$. We will show that the lemma property holds for cut $(S, V \setminus S)$ with high probability.

Let $E_{i, S} = \bigcup_{I \in F_{i, S}} E(I)$ and $k = |E_{i, S}|$. Consider a process where we sample each edge in $E_{i, S}$ individually and recalculate the value $f_{i, S}$ after each sample. Denote $Z_0$ as the initial value and $Z_k$ as the final value. It is easy to see that $Z_0 = f_{i, S}$, $Z_k = f'_{i, S}$ and that it is a martingale. 

Let $e_t$ be the edge sampled during step $t \in [k]$. Denote 
\[
\M^S(e_t) = \sum_{\substack{I\in F_{i, S}:\\ e_t \in E(I)}} w(I).
\]
Because we sample each edge with probability $\geq {2^{-1/(2r^*)}}$, $|Z_{t} - Z_{t - 1}| \leq \sqrt{2} M^S(e_t)$. On one hand, if $\wws_M(e_t) \geq \Upsilon$, the edge is not sampled and $Z_{t} = Z_{t-1}$. On the other hand, when $\wws_M(e_t) \leq \Upsilon$, using the fact that $2^i \leq k_I < 2^{i + 1}$ for $I \in F_{i, S}$, we can bound $M^S(e)$ as follows:
\[
    \Upsilon \geq \wws_M(e_t) = \sum_{\substack{I\in \MG GM:\\ e_t \in E(I)}} \frac{w(I)}{k_I} \geq \sum_{\substack{I\in F_{i, S}:\\ e_t \in E(I)}} w_\M(I)2^{-i-1} = M^S(e_t) 2^{-i-1}.
\]
Hence we can set $c_t = \sqrt{2} \min(\Upsilon 2^{i+1}, M^S(e_t))$, and $c_t \geq |Z_t - Z_{t - 1}|$ in both cases. Then we have
\[
    \sum_{t = 1}^{k} c_t^2 \leq 2 \Upsilon 2^{i+1} \sum_{t = 1}^{k} M^S(e_t) \leq 2 r^* \Upsilon f_{i, S} 2^{i+1}.
\]
Now let $\xi = \frac{\epsilon}{2} \max \left( \frac{e_{i,S} 2^{i - 1}}{\gamma \pi_i}, f_{i, S} \right)$. By \Cref{lem:azuma} and because $e_{i, S} \geq 2^j \pi_i$,
\begin{align*}
    \Pr(|Z_k - Z_0| \geq \xi) & \leq 2 \exp\left( \frac{-\xi^2}{2\sum_{t = 1}^{k} c_t^2} \right) \leq 2 \exp\left( \frac{-\eps^2 }{16 r^* \Upsilon f_{i, S} 2^{i+1}} \cdot f_{i, S} \cdot \frac{e_{i,S} 2^{i - 1}}{\gamma \pi_i} \right) \\
    & \leq 2 \exp\left( \frac{-\eps^2 2^j}{64 r^* \gamma \Upsilon} \right) \leq 2 \exp(-(d_1 + r + 2 r^*)2^j \ln n).
\end{align*}

Because instances in $F_{i, S}$ are $\pi_i$-heavy, \Cref{lem:proj_num} implies that the number of distinct sets $F_{i, S}$ is at most $n^{2 \cdot 2^j r^*}$. Using a union bound over them, we get that the statement of the lemma holds with probability at least $1 - 4 n^{-d_1 + r}$.
\end{proof}

Now, because there is at most $n^r$ motif instances, there are at most $n^r$ non-empty sets $F_i$. Therefore, we can do a union bound over this quantity, which yields an overall probability of at least $1 - 4 n^{-d_1}$ for which the statement of the \Cref{lem:conn_bound_core} holds for all cuts. 

Finally, for all cuts $(S, V \setminus S)$, by the second property of \Cref{lem:conn_overlap},
\begin{align*}
    |\Val_{G', M}(S, V \setminus S) - \Val_{G, M}(S, V \setminus S)| & \leq |\sum_{i = 0}^\Lambda (f'_{i, S} - f_{i, S}) |
    \leq \frac{\eps}{2} \sum_{i=0}^{\Lambda} \max(\frac{e_{i, S}}{\gamma}, f_{i, S}) \\
    & \leq \frac{\eps}{2} \sum_{i=0}^{\Lambda} \frac{e_{i, S}}{\gamma} + f_{i, S} \\
    & \leq \eps \Val_{G, M}(S, V \setminus S),
\end{align*}
which implies that $G'$ is a $(M, \eps)$-motif cut sparsifier.
\end{proof}

Finally, the algorithm correctness follows by the fact that our sampling strategy conforms with requirements of \Cref{thm:iterative_conn_sampling}.

\begin{theorem} \label{thm:fast_spar_spars}
Let $\{ M_i \}_{i\in[l]}$ --- set of motifs, $G=(V,E,w)$ --- a directed weighted graph, and $\epsilon'$ be the inputs of the \textsc{FastPartialSparsification}. Then for each $i \in [L]$ with probability at least $1 - 5 n^{-d_1}$, the output is an $(\epsilon', M_i)$-motif cut sparsifier of $G$ for all $i \in [L]$ simultaneously.
\end{theorem}
\begin{proof}
Fix $M = M_i$.
The algorithm makes one call to \textsc{Connectivities} algorithm related to motif $M$, which we assume to have sucess probability at least $1 - n^{-d_1}$.
By \Cref{lem:wwse_approx} and \Cref{lem:motif_conn_approx}, $\wwsee_{M_i}(I) \leq 2 r^* \wws_{M_i}(I)$, therefore \Cref{alg:fast_par_spars} satisfies prerequisites of \Cref{thm:iterative_conn_sampling} with $\eps = \eps'$ and it's output is a $(\epsilon', M_i)$-motif cut sparsifier with probability at least $1 - 4 n^{-d_1}$. Hence, the final success probability is at least $1 - 5 n^{-d_1}$.
\end{proof}

\subsection{Size of sparsifier from \fparspar{}}

In this subsection, the sparsifier size.
More precisely, we bound the number of edges sampled with probability $1$, which are basically an analog of the critical edges from \Cref{alg:par_spars}. To do this, we will utilize a classic result on sum of inverse connectivities. We include the proof for completeness.

Because we iteratively apply \fparspar{} to the same graph multiple times, all other edges will be discarded with high probability, yielding us a sparsifier size bound as detailed in \Cref{lem:final_size}.

\begin{lemma}[Corollary of Lemma 6.9 of \cite{chekuri2018minimum}] \label{lem:conn_sum_bound}
Let $H=(V, F, w)$ be a weighted hypergraph. For $I \in F$, let $k_I$ denote the hyperedge connectivities and $\kappa_I$ denote the hyperedge strengths. Then
\[
    \sum_{I \in F} \frac{w(I)}{k_I} \leq \sum_{I \in F} \frac{w(I)}{\kappa_I} \leq n - C
\] where $C$ is the number of connected components in $H$.
\end{lemma}
\begin{proof}
Since $k_I \geq \kappa_I$ for all $I \in F$, it is enough to only prove the second inequality.

Let $H'=(V, E, w')$ where $w'(I) = \frac{w(I)}{\kappa_I}$. Let $(S, V \setminus S)$ be the minimum cut in $H$. Notice that for all $I \in F$ that cross this cut, $\kappa_I$ is equal to the size of this cut, hence the size of this cut in $H'$ is $1$. On the other hand, because $\kappa_I$ is not higher than the size of any cut crossing $I$, the size of each cut in $H'$ is at least $1$, hence the size of the minimum cut in $H'$ is $1$.

Now, we will show the claim by induction on $C$. If $n - C = 0$, the claim holds trivially. Otherwise, we assume that the claim holds for all hypergraphs with a bigger number of connected components.

Find a minimum cut in $H'$ and remove all the cut hyperedges from the graph $H$. Let $J=(V, F', w)$ be the resulting graph, and let $J'=(V, F', w'')$ be its reweighted version. This increases the number of connected components, therefore by induction the inequality holds for the new graph $J$. Because removal of edges can only decrease the strengths of edges, $w''(I) \geq w'(I)$ for all $I \in F'$. Therefore,
\[
    \sum_{I \in F} \frac{w(I)}{\kappa_I} \leq \sum_{I \in F'} w''(I) + \sum_{I \in F: I \text{ crosses } (S, V \setminus S)} w'(I) \leq n - C - 1 + 1 = n - C. \qedhere
\]
\end{proof} 

\begin{lemma} \label{lem:fast_spar_size}
Let $\{ M_i \}_{i\in[l]}$ --- set of motifs, $G=(V,E,w)$ --- a directed weighted graph, and $\epsilon' \geq 0$ be the inputs of the \textsc{FastPartialSparsification}. The number of edges sampled with probability $1$ in the algorithm is at most 
\[
    \sum_{i = 1}^{L} \frac{256 (n - 1) (d_1 + r_i + 2r_i^*) (r_i^*)^4 r_i \log n \ln n}{{\eps'}^2 }.
\]
\end{lemma}
\begin{proof}
Fix $i \in [i]$. By \Cref{lem:wwse_approx} and \Cref{lem:motif_conn_approx}, $\wwsee_{M_i}(I) \leq 2 r_i^* \wws_{M_i}(I)$. Let $\tau$ be the number of edges sampled with probability $1$ which are added to $E_+$ when considering motif $M_i$ in line~\ref{line:add_crit_weak} of \Cref{alg:fast_par_spars}. We have
\[
    \tau \cdot \Upsilon' \leq \sum_{e \in E} \wwsee_{M_i}(e) \leq 2 r^*_i \sum_{e \in E} \wws_{M_i}(e) \leq 2 (r_i^*)^2 \sum_{I \in \MG GM} \frac{w(I)}{k_I} \leq 2(r_i^*)^2 (n - 1)
\]
where the last inequality follows by \Cref{lem:conn_sum_bound}. Hence $\tau \leq 2 (r_i^*)^2 (n - 1) / \Upsilon'$.

We obtain the bound by summing over all $i$.
\end{proof}

\subsection{Running time of \fparspar{}}

Remember that the main practical difference between \parspar{} and \fparspar{} is in the running time, which we show in this subsection.

Recall that $W = \max_{e \in E} w(e) / \min_{e \in E} w(e)$.
\begin{lemma} \label{lem:fast_spar_time}
Let a directed weighted graph $G=(V, E, w)$, $\eps' > 0$ and a motif set $\{M_i\}_{i=1}^{L}$ be the input of $\fparspar{}$. Then its running time is bounded by
\[
\tO(L (n^2 + (r^r + n^{\omega \lceil r / 3\rceil } + n^{2\lceil r/3 \rceil}) \log W)).
\]
\end{lemma}
\begin{proof}
First consider the loop at line~\ref{line:fast_spar_loop}. In this loop, \textsc{MotifWeights} is called $O(\sum_{i \in [L]} \Lambda_{M_i})$ times, \textsc{Connectivities} is called $L$ times, $\wwsee_{M_i}(e)$ is calculated for all $e \in E$ and $i \in [L]$, and the set $E_+$ is updated $L$ times.

First we bound $\Lambda_{M_i}$. Since $\Lambda_{M_i} = \lceil \log \max_{e \in E} k_{M_i, e} / \min_{e \in E} k_{M_i, e} \rceil$, and $\min_{e \in E} k_{M_i, e} \geq \min_{e \in E} w(e)^{r_i^*}$ and $\max_{e \in E} k_{M_i, e} \leq n^{r_i} \max_{e \in E} w(e)^{r_i^*}$, we have $\Lambda_{M_i} \leq r_{max}^* \log W + r_{max} \log n + 1$.

The time needed to calculate $\wwsee_{M_i}$ for all $e \in E$ is $O(\Lambda_{M_i} |E|)$, given weights $w_{M_i, j}$, $j \in [\Lambda_{M_i}]$. Hence, by \Cref{thm:motif_counting} and \Cref{thm:all_pairs_conn}, the running time of this segment is
\begin{align*}
    &\tO(Ln^2) + O((r_{max}^* \log W + r_{max} \log n) L (r^* r^r + n^{\omega \lceil r / 3\rceil } + r^2 n^{2\lceil r/3 \rceil}) ) \\
&= \tO(L (n^2 + (r^r + n^{\omega \lceil r / 3\rceil } + n^{2\lceil r/3 \rceil}) \log W)).
\end{align*}
Since the rest can be done in time $O(|E|)$, the first part dominates the runtime.
\end{proof}

\subsection{\motspar{} with \fparspar{}}
Similarly to \parspar{}, the final algorithm is obtained by running the \motspar{}. In the next theorem we derive its properties.

\begin{proof}[Proof of \cref{thm:intro-main-general-2}]
The runtime follows from \Cref{lem:fast_spar_time} and \Cref{thm:final_time}. Notice that $W$ multiplies by at most $n^{c_1}$ during the execution of \motspar{}, since each weight is multiplied by at most $2^{-1/(2r^*)}$ each iteration.

The sparsifier size follows from \Cref{lem:final_size} and \Cref{lem:fast_spar_time}. More precisely, it is at most
\[
O\left(L \frac{n r_{max} (r_{max}^*)^7 \log^4 n}{{\eps}^2 }\right).
\]

The probability follows from \Cref{thm:final_spars} and \Cref{thm:fast_spar_spars} by setting $d_1 = c_1 + 1$.
\end{proof}

%% file: lower_bound.tex
\section{Lower Bound for Induced Motif Sparsification}\label{sec:lower-bound}

\subsection{Overview}

In contrast to the rest of the paper, in this section we consider the question of motif-cut sparsiciation in the context of induced motifs. That is, unlike in the rest of the paper, a subgraph is only considered to be a motif instance if it is an induced subgraph.

\begin{definition}
Let $G=(V, E)$ be a directed graph, and let $M=(V_M, E_M)$, a weakly connected directed graph, be our motif. An {\bf induced} subgraph of $G$ that is \emph{isomorphic} to $M$ is considered to be an {\bf induced} motif instance. The set of all {\bf induced} instances of $M$ is $G$ is denoted $\overline{\mathcal M}(G, M)$ (with the overline differentiating it from the set of not-necessarily-induced motif instance $\mathcal M(G, M)$).
\end{definition}

This can be simply generalized to undirected graphs and motifs, as described in \cref{sec:prelims}.

\begin{definition}
We extend the definitions of the \emph{weigh of a motif instance}, the \emph{size of a motif cut} and the concept of an $(M, \epsilon)$\emph{-motif cut sparsifier} analogously from \Cref{def:motif-weight,def:motif-cut-size,def:motif-cut-sparsifier}, with the exception that we denote the motif size of a cut by $\OVal_{M, G}$.
\end{definition}

In this section, we rule out the possibility of constructing any non-trivial induced-motif-cut sparsifiers in full generality, by demonstrating an example of a graph and a motif where this is not possible:

\lowerbound*

In the rest of the section, we recall our lower-bound construction from \cref{sec:lower-bound-overview}, give an overview of our proof, then finally prove \cref{thm:graphlet-main} formally in \cref{sec:graphlet-proof}.

\paragraph{Construction:}
Our input graph will be the undirected, unweighted clique with the three edges of a specific triangle $(a,b,c)$ removed. More formally, we define $\Delta^-=(V, E_\Delta^-)$ as an unweighted, undirected graph on $n$ vertices, where
$$E_\Delta^-=\binom{V}{2}\setminus\big\{\{a,b\},\{b,c\},\{c,a\}\big\},$$
for distinct special vertices $a, b, c\in V$.

We call these three special vertices the central triangle, and all other vertices the periphery. Our motif will be the induced undirected $2$-path -- i.e. $3$ vertices with exactly $2$ edges between them.


\paragraph{Proof Sketch:}
Note first the distribution of $2$-path motifs in $\Delta^-$: We have exactly $3(n - 3)$ motifs, each having two vertices in the central triangle and one in the periphery. Suppose a weighted graph $\widehat G=(V,\widehat E, w)$ approximates the induced-motif-cut structure of $\Delta^-$ to within a $(1\pm\epsilon)$-factor for some $\epsilon=\Omega(1)$. We assume such a $\widehat G$ exists and, through a series of claims, we show that $\widehat G$ must necessarily be dense.

First we show that nearly all induced $2$-paths in $\widehat G$ must contain one vertex from the periphery and two from the central triangle -- similarly to how it is in $\Delta^-$ (\cref{claim:m-2-graphlets}). Next, we show that most peripheral vertices must have induced $2$-paths in common with all three central vertices (\cref{claim:graphlet-strong-constribution}). This implies that most periphery vertices must have a heavy edge (of weight $\Omega(1)$) connecting them to at least one of the central vertices (\cref{claim:graphlet-strong-connection}). (This statement may seem trivial at first glance, but is actually the crux of the proof;~\cref{ex:graphlet-sparsifier} shows a similar construction where the analogous statement is false, leading to a valid sparsifier.) Finally, we argue that at least one of the central vertices must have $\Omega(n)$ heavy edges adjacent on it. This leads to $\Omega(n^2)$ \emph{not necessarily induced} $2$-paths; in order for most of these to not be induced, $\widehat G$ must be dense.

In what follows, we formalize the above argument, and show that any graph $\wh G$ approximating the induced-motif-cut structure of $\Delta^-$ to within a constant multiplicative error must have $\Omega(n^2)$ edges.

\begin{remark}
Throughout the proof we assume that the sparsifier $\wh G$ is undirected. Since our motif is also undirected this is without loss of generality: Indeed, for $u,v\in V$ we can replace any directed edges $(u,v)$ of weight $w_1$ and $(v,u)$ of weight $w_2$ by a single undirected edge $\{u,v\}$ of weight $w_1\cdot w_2$. Similarly, if exactly one of $(u,v)$ and $(v,u)$ is present, we can replace it with an undirected edge $\{u,v\}$ of weight $0$, without affecting the induced $P_2$ motif-graph. Thus the existence of a directed sparsifier implies the existence of an undirected sparsifier of equal or smaller size.
\end{remark}

\subsection{Proof of~\cref{thm:graphlet-main}}\label{sec:graphlet-proof}

\begin{proof}[Proof of~\cref{thm:graphlet-main}]

    We take $G=\Delta^-$ on $n$ vertices, and the motif which is the induced $2$-path ($P_2$) as our example. We may assume without loss of generality that $\epsilon \ge 100/n$. Suppose $\widehat G=(V, \widehat E, w)$ is a graph which well approximates the induced-motif-cut structure of $\Delta^-$, that is, for all cuts $S\subseteq V$
   \begin{equation}
       (1 - \epsilon) \OVal_{P_2, \Delta^-}(S, V \setminus S) \leq \OVal_{P_2, \widehat{G}}(S, V \setminus S) \leq (1 + \epsilon) \OVal_{P_2, \Delta^-}(S, V \setminus S).\label{eq:graphlet-main}
   \end{equation} 
   
  \begin{claim}\label{claim:total-graphlets}
    The total weight of induced $2$-paths motifs in $\widehat G$ is at most $3(1 + \epsilon)n$.
  \end{claim}

\begin{proof}
   We can estimate the weight of motifs in $\wh G$ by applying~\cref{eq:graphlet-main} to each singleton-cut in turn. This gives as that each vertex in the central triangle ($a$, $b$, and $c$) has at most $2(n-3)\cdot(1+\epsilon)$ motifs containing it. Similarly, the vertices in the periphery ($V\setminus\{a, b, c\}$) each have at most $3(1+\epsilon)$ motifs containing each. Since each motif contains exactly $3$ vertices, this is a total of at most
   $$\frac{3\cdot(2(n-3)\cdot(1+\epsilon)) + (n-3)\cdot(3(1+\epsilon))}3=(1+\epsilon)\cdot (2(n-3) + (n-3))\le3(1+\epsilon)n$$
   weight among all motifs.
\end{proof}

We categorize the motifs based on the number of central vertices they contain: $\mathcal M_i$ contains motifs with exactly $i$ central vertices and exactly $3-i$ vertices from the periphery. Hence, we have the partition
$$\overline{\mathcal M}(P_2,\wh G)=\mathcal M_0\cup\mathcal M_1\cup\mathcal M_2\cup\mathcal M_3.$$

We prove that all but a diminishingly small fraction of the motifs reside in $\mathcal M_2$.
\begin{claim}\label{claim:m-2-graphlets}
The total motif-weight of $\mathcal M_0\cup\mathcal M_1\cup\mathcal M_3$ is at most $21\epsilon n$, that is (with slight abuse of notation)
   $$w(\mathcal M_0)+w(\mathcal M_1)+ w(\mathcal M_3)\le27\epsilon n.$$
\end{claim}

\begin{proof}
First, consider~\cref{eq:graphlet-main} with the cut $S=\{a, b, c\}$. In $\Delta^-$ this cuts all motifs, therefore
$$3(n -3) = \OVal_{P_2, \Delta^-}(S, V\setminus S)\le(1-\epsilon)^{-1}\OVal_{P_2,\wh G}(S, V\setminus S)=(1-\epsilon)^{-1}\cdot w\left(\mathcal M_1 + \mathcal M_2\right),$$
since motifs in $\mathcal M_0$ and $\mathcal M_3$ don't cross this cut in $\wh G$. By~\cref{claim:total-graphlets}, this implies that the total weight of $\mathcal M_0$ and $\mathcal M_3$ is at most $4\epsilon n$. (Recall that $\epsilon \ge100/n$.)

Next, consider again~\cref{eq:graphlet-main} for each singleton cut containing the vertices $a$, $b$, and $c$ in turn. Similarly to the proof of~\cref{claim:total-graphlets}, this gives us that each central vertex has at least $2(n-3)\cdot(1-\epsilon)$ motifs containing it. To account for these, we must have
\begin{align*}
    3\cdot2(n-3)\cdot(1-\epsilon)&\le3w(\mathcal M_3)+2w(\mathcal M_2)+w(\mathcal M_1)\\
    &\le 2w(\overline{\mathcal M}(P_2,\wh G)) + w(\mathcal M_3) - w(\mathcal M_1)\\
    &\le6n(1+\epsilon) + 4\epsilon n- w(\mathcal M_1),
\end{align*}
by~\cref{claim:total-graphlets}. Therefore, $w(\mathcal M_1)\le6n(1+\epsilon) + 4\epsilon n - 6(n - 3)(1-\epsilon)\le 17n\epsilon$, which concludes the proof of the claim.
\end{proof}

Given that most motifs are in $\mathcal M_2$ we focus on these, and further partition them into $\mathcal M_{ab}$, $\mathcal M_{bc}$, and $\mathcal M_{ca}$ which respectively contain exactly $(a,b)$, $(b,c)$, and $(c,a)$ from the central triangle. The remaining motifs make up $\mathcal M_-=\mathcal M_0\cup\mathcal M_1\cup\mathcal M_3$, and their total weight is diminishingly small. We will prove that the total weight of the motifs in each of the main categories ($\mathcal M_{ab}$, $\mathcal M_{bc}$, and $\mathcal M_{ca}$) are roughly the same, that is roughly $n$. (The claim is phrased in terms of pairs of categories, as this will be the most useful form later on).

\begin{claim}\label{claim:graphlets-balance}$\mathcal M_{ab}$, $\mathcal M_{bc}$, and $\mathcal M_{ca}$ satisfy the following inequalities:
\begin{align*}
    &w(\mathcal M_{ab}) + w(\mathcal M_{bc}) \le 2n + 2\epsilon n,\\
    &w(\mathcal M_{bc}) + w(\mathcal M_{ca}) \le 2n + 2\epsilon n,\\
    &w(\mathcal M_{ca}) + w(\mathcal M_{ab}) \le 2n + 2\epsilon n.
\end{align*}
\end{claim}

\begin{proof}
Again, it suffices to look at~\cref{eq:graphlet-main} where $S$ is the singleton cut of a central vertex, say $a$. Such a cut in $\Delta^-$ contains exactly two thirds of the motifs, that is $2(n-3)$; in $\wh G$, this cut crosses all of $\mathcal M_{ab}$, all of $\mathcal M_{ca}$, none of $\mathcal M_{bc}$, and some subset of $\mathcal M_-$. Hence, $2(n-3)\cdot(1+\epsilon) \ge w(\mathcal M_{ab}) + w(\mathcal M_{ca})$. The other two claims hold by an identical argument.
\end{proof}

We now consider the behavior of peripheral vertices, that is vertices other than $a$, $b$, or $c$. We know that each peripheral vertex is contained in approximately $3$ motifs (by weight). In the original graph $\Delta^-$, each peripheral vertex contributed to each of the the categories $\mathcal M_{ab}$, $\mathcal M_{bc}$, and $\mathcal M_{ca}$ in equal measure. We show that the situation is approximately the same in $\wh G$. We say that a peripheral vertex $x$ contributes strongly to $\mathcal M_{ab}$ if a motif is supported on ${a, b, x}$ in $\wh G$, and it has motif-weight at least $1/2$. We define strong contribution analogously for $\mathcal M_{bc}$ and $\mathcal M_{ca}$.

\begin{claim}\label{claim:graphlet-strong-constribution}
    At least half of the peripheral vertices strongly contribute to \textbf{each of} $\mathcal M_{ab}$, $\mathcal M_{bc}$, and $\mathcal M_{ca}$.
\end{claim}

\begin{proof}
    Suppose for contradiction that this is not the case, and there are at least $(n - 3)/2$ vertices which \emph{do not} contribute strongly to at least one of the categories. By the pigeon-hole principle, at least $(n-3)/6$ vertices do not contribute strongly to a specific one of these categories - without loss of generality, we may assume that this is $\mathcal M_{ab}$. That is, there is a set $T$ of peripheral vertices where $|T|\ge(n-3)/6$ and no $x\in T$ contributes strongly to $\mathcal M_{ab}$.
    
    We now consider~\cref{eq:graphlet-main} for the cut $S=T\cup\{c\}$. Consider this cut in $\Delta^-$: It crosses all motifs containing $c$, but of the motifs containing $a$ and $b$, it crosses only $|T|$ of them. Hence it has a total size of $2(n-3) + |T|$. Now, consider this cut in $\wh G$: It crosses all motifs in $\mathcal M_{bc}$ and $\mathcal M_{ca}$, as well as some subset of the motifs in $\mathcal M_-$. By definition of $T$, motifs in $\mathcal M_{ab}$ contribute only at most $|T|/2$ to this cut. Therefore
$$
    (1-\epsilon)\cdot\big(2(n-3)+|T|\big)\le w(\mathcal M_{bc}) + w(\mathcal M_{ca}) + w(\mathcal M_-) + |T|/2.
$$
Applying~\cref{claim:graphlets-balance} and~\cref{claim:m-2-graphlets} we get that
$$(1-\epsilon)\cdot(2(n-3) + |T|)\le(2n + 2\epsilon n) + 21\epsilon n + |T|/2.$$
Hence, $(1/2-\epsilon)\cdot|T|\le26\epsilon n$, which contradicts our assumption that $|T|\ge(n-3)/6$ since $\epsilon <1/500$.
\end{proof}

Thus we have at least $(n-3)/2$ peripheral vertices which strongly contribute to all three of $\mathcal M_{ab}$, $\mathcal M_{bc}$, and $\mathcal M_{ca}$. We show that any such vertex must be strongly connected to at least one of the central vertices (that is connected by an edge of weight at least $1/\sqrt{2}$).

\begin{claim}\label{claim:graphlet-strong-connection}
Suppose $x\in V\setminus\{a, b, c\}$ strongly contributes to each of $\mathcal M_{ab}$, $\mathcal M_{bc}$, and $\mathcal M_{ca}$. Then at least one of $\{a,x\}$, $\{b, x\}$, or $\{c,x\}$ exists in $\wh E$ with weight at least $1/\sqrt{2}$.
\end{claim}

\begin{proof}
By assumption, each of $\{a, b, x\}$, $\{b, c, x\}$ and $\{c, a, x\}$ is the support of an induced $2$-path motif. Therefore, in $\wh G$, $x$ must be connected to $a$ or $b$, as well as $b$ or $c$, as well as $c$ or $a$. Overall, $x$ is connected to at least two of the central vertices -- without loss of generality we may assume that these are $a$ and $b$. We know that $\{a, b, x\}$ is the support of a motif in $\wh G$ -- we now know that this motif must be $a-x-b$, that is $\{a, b\}\not\in\wh E$. We further know that the weight of the $a-x-b$ motif, that is $w(\{a,x\})\cdot w(\{b,x\})\ge1/2$. Hence, at least one of these weights is at least $1/\sqrt{2}$, as claimed.
\end{proof}

Finally, we finish the proof of~\cref{thm:graphlet-main}, by showing that there are a large number of \emph{not-necessarily-induced} $2$-paths in $\wh G$, each of weight at least $1/2$. By~\cref{claim:graphlet-strong-connection} and~\cref{claim:graphlet-strong-constribution} at least $(n-3)/2$ peripheral vertices are strongly connected to a central vertex. By the pigeon-hole principle, at least $(n-3)/6$ peripheral vertices are strongly connected to one specific central vertex; we may assume without loss of generality that this is $a$.

Let the set of peripheral vertices strongly connected to $a$ be $A\subseteq V\setminus\{a, b, c\}$ (where we know that $|A|\ge(n-3)/6$). For any pair of distinct vertices $x, y\in A$, $x-a-y$ constitutes a $2$-path in $\wh G$ of weight at least $1/2$. A $2$-path like this is not necessarily induced, however, for it not to be induced, $\{x,y\}$ must be in $\wh E$.

Suppose for contradiction that $|\wh E|\le n^2/200$. Then, of all the $2$-paths in $A\times \{a\}\times A$, at least
\[\binom{|A|}2-\frac{n^2}{200}\ge\frac{n-3}6 \cdot\left(\frac{n-3}6-1\right) \cdot \frac{1}2-\frac{n^2}{200}\ge \frac{n^2}{100}\]
of them are actually induced, and therefore count as motifs. These motifs contribute to $\mathcal M_1$, and therefore the total weight of $\mathcal M_1$ is at least $1/2\cdot n^2/100$ contradicting~\cref{claim:m-2-graphlets}.

This shows that $|\wh E|$ is at least $n^2/200$, concluding the proof.
\end{proof}

%% file: appendix.tex
\appendix

\section{Auxillary Lemmas}
\label{appendix}

\begin{lemma}
\label{lem:exp_bounds}
Let $x, a, b \in \mathbb{R}$, $0 < a/x < 1$, $b > 0$, $ab < 1$, and $x \ge 1$. Then
\[
    \left(1 + \frac{a}{x} \right)^{bx} \leq e^{ab} \leq 1 + 2ab,
\]
\[
    \left(1 - \frac{a}{x} \right)^{bx} \geq  1 - ab.
\]
\end{lemma}
\begin{proof}
The first set of inequalities follows from the fact that $1+t \le e^t$ for any $t \in \R$ and $e^t \le 1+2t$ for $t \in [0,1]$, both of which follow from series expansion of $e^t$. The second inequality follows from the fact that $(1+t)^r \ge 1+tr$ for any $t \ge -1$ and $r \ge 0$.
\end{proof}